\documentclass[acmsmall,authorversion,nonacm]{acmart}
\AtBeginDocument{%
  \providecommand\BibTeX{{%
    \normalfont B\kern-0.5em{\scshape i\kern-0.25em b}\kern-0.8em\TeX}}}

\setcopyright{none}
\usepackage{booktabs}   
\usepackage{subcaption} 

\usepackage{wasysym}
\usepackage{xspace}
\usepackage{array,multirow}
\usepackage{adjustbox}

\usepackage{todonotes}

\usepackage{amsmath}
\usepackage{mathtools}
\usepackage{suffix}
\usepackage{tikz-cd}
\usepackage{mathpartir}
\usepackage{enumitem}
\usepackage{stmaryrd}
\usepackage[all]{xy}
\usepackage{twoopt}
\usepackage{array}
\usepackage{listings}
\usepackage{multirow, bigdelim}
\usepackage{graphicx}

\lstdefinelanguage{llql}%
{morekeywords={
  if,then,else,let,in,
  op2,op1,map,map2,foldl,
  scanl,scanr,map3,fun,reduce,
  shift1L, shift1R,
  op,var,J,
  sum,prod,dot,fst,
  zip,
  true,false,
  OnesLike,
  ZerosLike,
  pair,proj,cos,sin,exp
  },%
  sensitive,%
  morecomment=[l]//,%
  morecomment=[s]{/*}{*/},%
  morestring=[b]",%
  morestring=[b]',%
  showstringspaces=false,%
  morecomment=[s][\color{gray}]{@}{\ },%
    breaklines=true,%
  mathescape=true,%
showspaces=false,
showtabs=false,
showstringspaces=false,
breakatwhitespace=true,
  aboveskip=1pt,
  belowskip=1pt,
  lineskip=-0.2pt,
  numbersep=5pt,
  numberstyle=\tiny\ttfamily,
  basicstyle=\small\ttfamily,
  keywordstyle=\bfseries\color{blue!70!black},%
  columns=fullflexible,
  frame=single,
  escapeinside={(*@}{@*)},
  literate={->}{$\rightarrow\;$}{2}
           {<}{$\langle$}{1}
           {>}{$\rangle$}{1}
}[keywords,comments,strings]%

\lstMakeShortInline[columns=fixed]!

\newenvironment{absolutelynopagebreak}
  {\par\nobreak\vfil\penalty0\vfilneg
   \vtop\bgroup}
  {\par\xdef\tpd{\the\prevdepth}\egroup
   \prevdepth=\tpd}

\DeclareMathOperator{\J}{J}

\newcommand{\supfull}{\CIRCLE}
\newcommand{\suphalf}{\LEFTcircle}
\newcommand{\supnone}{\Circle}

\newcommand{\RR}{\mathbb{R}}

\newtheorem{notation}{Notation}

\newcommand{\codekw}[1]{\texttt{\small\bfseries\color{blue!70!black}#1}}

\newcommand{\Diff}{\mathbf{Diff}}
\newcommand{\sem}[1]{\llbracket #1\rrbracket}

\newcommand{\defeq}{\stackrel {\mathrm{def}}=}

\makeatletter
\newcommand\@TyAlph[1]{%
\ifcase #1\or \tau\or \sigma\or \rho\else \@ctrerr \fi%
}
\newcommand\ty[1][1]{{\@TyAlph{#1}}}

\newcommand\tvar[1][1]{{\@TyVarAlph{#1}}}
\newcommand\@TyVarAlph[1]{%
\ifcase #1\or \alpha\or \beta\or \gamma\else \@ctrerr \fi%
}

\newcommand\var[1][1]{{\@VarAlph{#1}}}
\newcommand\@VarAlph[1]{%
\ifcase #1\or x\or y\or z\or u\or v\or w\else \@ctrerr \fi%
}

\newcommand\trm[1][1]{{\@TermAlph{#1}}}
\newcommand\@TermAlph[1]{%
\ifcase #1\or t\or s\or r\else \@ctrerr \fi%
}

\newcommand\val[1][1]{%
\ifcase #1\or v\or w\or u\else \@ctrerr \fi%
}

\newcommand\op[1][1]{%
\ifcase #1\or \mathsf{op}\or \mathsf{op}'\or \mathsf{op}''\else \@ctrerr \fi%
}
\makeatother

\newcommand\Dsynrevsymbol[1][]{{\scalebox{0.8}{$\overleftarrow{\mathcal{D}}$}_{#1}}}

\newcommand\UNFSymbol{\mathbf{UNF}}

\newcommand\invUNFSymbol{\mathbf{UNF}^{-1}}

\newcommand\cnst{\underline{c}}

\newcommand\reals{\mathbf{R}}

\newcommand{\BB}{\mathbb{B}}

\newcommand\ArraySym{\mathbf{A}}
\newcommand\Array[2]{\ArraySym[#1]^{#2}}

\newcommand{\grammarcomment}[1]{\textit{\small #1}}

\newcommand{\Source}{\mathbf{Source}}
\newcommand{\Target}{\mathbf{Target}}

\newcommand{\cost}{Cost}
\newcommand{\transto}{\text{ $\leadsto$ }}
\newcommand{\directD}[3]{\Dsynrevsymbol^{#1}_{#2;#3}}
\newcommand{\ad}[2]{\mathcal{D}_{#1;#2}}
\newcommand{\grad}{\nabla}

\newcommand{\pcomp}{\widetilde{;}}
\newcommand{\comp}{;}

\newcommand{\icomp}{\circ}

\newcommand{\SynSource}{SynSource}
\newcommand{\SynTarget}{SynTarget}
\newcommand{\catC}{\mathcal{C}}
\newcommand{\ConcatS}{Concat1}
\newcommand{\ConcatT}{Concat2}
\newcommand{\seml}{\llbracket}
\newcommand{\semr}{\rrbracket}

\newcommand{\concatcomp}{\mathmakebox[2ex][s]{+\kern-1ex+\kern0.8ex}}
\newcommand{\CartSp}{\mathbf{CartSp}}

\newcolumntype{R}[2]{%
    >{\adjustbox{angle=#1,lap=\width-(#2)}\bgroup}%
    l%
    <{\egroup}%
}
\newcommand*\rot{\multicolumn{1}{R{90}{0em}|}}
\newcommand{\plots}[1]{\mathcal{P}_{#1}}
\newcommand\seq[2][]{\left(#2\right)_{#1}}
\newcommand{\CSource}{\mathcal{C}_{Source}}
\newcommand{\CTarget}{\mathcal{C}_{Target}}
\newcommand{\NestA}{NAO}
\newcommand{\smartpara}[1]{\noindent \textbf{#1.}}
\newcommand{\pluseq}{\mathrel{+}=}

\begin{document}

\title{Denotationally Correct, Purely Functional, Efficient Reverse-mode Automatic Differentiation}

\author{Mathieu Huot}
\email{mathieu.huot@cs.ox.ac.uk}   
\affiliation{
  \institution{University of Oxford}
  \city{Oxford}
  \country{UK}}

\author{Amir Shaikhha}
\affiliation{
  \institution{University of Edinburgh}
  \city{Edinburgh}
  \country{UK}}

\renewcommand{\shortauthors}{Huot, et al.}

\begin{abstract}
    Reverse-mode differentiation is used for optimization, but it introduces references, which break the purity of the underlying programs, making them notoriously harder to optimize.
    We present a reverse-mode differentiation on a purely functional language with array operations.  
    It is the first one to deliver a provably efficient, purely functional, and denotationally correct reverse-mode differentiation.
    We show that our transformation is semantically correct and verifies the cheap gradient principle. 
    Inspired by PROPs and compilation to categories, we introduce a novel intermediate representation that we call `unary form'.
    Our reverse-mode transformation is factored as a compilation scheme through this intermediate representation.
    We obtain provably efficient gradients by performing general partial evaluation optimizations after our reverse-mode transformation, as opposed to manually derived ones.
    For simple first-order programs, the obtained output programs resemble static-single-assignment (SSA) code. 
    We emphasize the modularity of our approach and show how our language can easily be enriched with more optimized primitives, as required for some speed-ups in practice.
\end{abstract}


\maketitle

\section{Introduction}
\label{sec:intro}



Deep learning is moving towards increasingly sophisticated optimization objectives that employ tensors and operations on tensors.
Reverse-mode Automatic Differentiation (AD) is a technique to automatically compute the gradient of objective functions of the form $\RR^n\to\RR$.
Such functions appear a lot in practice: for instance, as loss functions in machine learning.

In order to reach the efficiency of the usual imperative version of reverse-mode, the transformations usually introduce references, even in functional languages \cite{lantern_icfp}.
The lack of purity in reverse-mode  makes it significantly harder to optimize and parallelize. Sophisticated heuristics are often used (e.g. \cite{xla}), which provide no theoretical performance guarantee.
As a result, to optimize for efficiency, a specific hand-crafted reverse-derivative must often be given for every non-elementary operation, even if an automated one can be compositionally obtained from the derivatives of its elementary constituting operations.
Abstracting away from imperative code in automatic differentiation is still a hurdle that functional implementations need to overcome.
 
In this paper, we define a purely functional (without references or control mechanisms such as state monads), denotationally correct, and provably efficient reverse-mode AD. 
To do so, we define the Unary Normal Form (UNF) representation inspired by PROPs~\cite{maclane1965categorical} 
and compilation to categories~\cite{elliott2017compiling}.
We can easily define and prove correctness of reverse-mode on this representation. 
The whole reverse-mode transformation is obtained by compiling the language to this Intermediate Representation (IR), applying the simpler reverse-mode transformation, and compiling again to the original language.
After standard optimizations, the output program looks like SSA~\cite{cytron1989efficient} or ANF~\cite{sabry1993reasoning}, which leads to more efficient implementations.

\begin{table}
 \label{fig:comparison-table}
 \small
\addtolength{\tabcolsep}{-1pt}
 \begin{tabular}{|l|c|c|c|c|c|c|c|c|c|c|c|c|c|c|c|c|c|c|}
 \hline
  & \rot{This Paper} & \rot{\cite{lantern_icfp}}  & \rot{\cite{shaikhha2019efficient}}  & \rot{\cite{huot2020correctness}}  & \rot{\cite{brunel2019backpropagation}}  & \rot{\cite{abadi-plotkin2020}}  & \rot{\cite{barthe2020versatility}}  & \rot{\cite{pearlmutter2008reverse}}  & \rot{\cite{Elliott:2018:SEA:3243631.3236765}} & \rot{\cite{sherman2021}}  & \rot{\cite{vytiniotis2019differentiable}}  & \rot{\cite{mak2020differential}}  & \rot{\cite{vakar2020reverse}}  & \rot{\cite{Manzyuk2012}}  & \rot{\cite{cockett2019reverse}}  & \rot{\cite{gallagher-sdg}} 
  & \rot{\cite{krawiec2022provably}}
  & \rot{\cite{paszke2021parallelism}}
  \\ \hline
Reverse Mode &
\supfull & \supfull & \supnone & \supfull & \supfull & \supfull & \supnone & \supfull & \supfull & \supnone & \supfull & \supfull & \supfull & \supnone & \supfull & \supnone & \supfull & \supfull\\ \hline
Complexity &
\supfull & \suphalf & \supnone & \supnone & \suphalf & \suphalf & \supnone & \supfull & \supnone & \supnone & \suphalf & \supnone & \supnone & \supnone & \supnone & \supnone & \supfull & \supfull \\ \hline
Pure Derivatives &
\supfull & \supnone & \supfull & \supfull & \supfull & \supfull & \supfull & \supnone & \supfull & \supfull & \supfull & \supfull & \supfull & \supfull & \supfull & \supfull & \supfull & \supfull \\ \hline
Correctness &
\supfull & \supnone & \supnone & \supfull & \supfull & \supfull & \supfull & \supnone & \supfull & \supfull & \supnone & \supfull & \supfull & \supfull & \supfull & \supfull & \supfull & \supnone \\ \hline
Tensor Support &
\supfull & \supfull & \supfull & \supnone & \supnone & \supnone & \supnone & \supnone & \supnone & \supnone & \supfull & \supnone & \supfull & \supnone & \supnone & \supnone 
& \supnone & \supfull
\\ \hline
HO Functions &
\supnone & \supfull & \supfull & \supfull & \supfull & \supnone & \supfull & \supfull & \supnone & \supfull & \supfull & \supfull & \supfull & \supfull & \supnone & \supfull 
& \supfull & \supnone\\ \hline
Recursion &
\suphalf & \supfull & \suphalf & \suphalf & \supnone & \supfull & \supnone & \supfull & \supnone & \suphalf & \supnone & \supnone & \supnone & \supnone & \supnone & \supnone 
& \supnone & \supnone\\ \hline
Conditional &
\supfull & \supfull & \supfull & \supfull & \supnone & \supfull & \supfull & \supfull & \supnone & \supfull & \supnone & \supnone & \supnone & \supnone & \supnone & \supnone 
& \supnone & \supnone \\ \hline
 \end{tabular}
 \caption{Comparison of different functional differentiable programming frameworks.
 $\supfull$ means that the property is verified, and $\supnone$ means that it is absent in the work.
 $\suphalf$ for complexity means that the proof is not fully covered, and for recursion, that it does not support general recursion but map, reduce and/or fold. HO stands for higher-order. Correctness is ticked if a proof is formalized in the paper.}
 \label{tbl:relwork}
 \end{table}

\subsection{Examples}

We introduce the general idea of efficient reverse-mode in a functional setting through the following examples.

\begin{example}[First-order term]
Let us consider the term !let w$_1:\reals$ = x$_1$ * x$_2$ in let w$_2:\reals$ = w$_1$ * x$_1$ in w$_2$! in the context $\Gamma:=\{x_1:\reals,x_2:\reals, x_3:\reals\}$.\\
After an (inefficient) reverse-mode transformation, we obtain:
\begin{center}
    \begin{tabular}{l}
        !let w$_1:\reals$,w$_1':\reals^4\to\reals^3$! = !< x$_1$ * x$_2$, fun (y$_1$,$\ldots$, y$_4$) -> (y$_1$+x$_2$*y$_4$, y$_2$+x$_1$*y$_4$, y$_3$)> in!\\
        !let w$_2:\reals$,w$_2':\reals^5\to\reals^3$! = !< w$_1$*x$_1$, fun (y$_1$,$\ldots$, y$_5$) -> w$_1'$!!(y$_1$+w$_1$*y$_5$, y$_2$, y$_3$, y$_4$+x$_1$*y$_5$)>!\\
        !in w$_2'$(0,0,0,0,1)!
    \end{tabular}
\end{center}
The part !(0,0,0,0,1)! corresponds to initializing the tangent variables in the imperative reverse-mode algorithm. 
After some general partial evaluation techniques that will be detailed further in the paper, we obtain:    

    \begin{center}
            \begin{tabular}{l}
                !let w$_1:\reals$ = x$_1$ * x$_2$ in!\\ 
                !let w$_2:\reals$ = w$_1$ * x$_1$ in!\\
                !let y$_1:\reals$,y$_2:\reals$,y$_3:\reals$,y$_4:\reals$,y$_5:\reals$ = 0,0,0,0,1 in!\\
                !let y$_1':\reals$! != y$_1$+w$_1$*y$_5$ in!\\
                !let y$_4':\reals$! != y$_4$+x$_1$*y$_5$ in!\\
                !(y$_1'$+x$_2$*y$_4'$!!, y$_2$+x$_1$*y$_4'$!!, y$_3$)!
            \end{tabular}
    \end{center}   
This is very close to the SSA form \cite{cytron1989efficient} of what the imperative reverse-mode differentiation of our initial term would be.

\begin{absolutelynopagebreak}
This term can be further optimized via constant propagation and algebraic simplifications to give
    \begin{center}
            \begin{tabular}{{c}}
                !let w$_1:\reals$ = x$_1$ * x$_2$ in!\\ 
                !let w$_2:\reals$ = w$_1$ * x$_1$ in!\\
                !(w$_1$+x$_2$*x$_1$, x$_1$*x$_1$, 0)!
            \end{tabular}
    \end{center}
\end{absolutelynopagebreak}

\end{example}

\begin{example}[Simple operations on arrays]
    On arrays, three simple operations of interest are the dot-product of two vectors, and the product or sum of the elements of a vector.
    In a functional setting, these can be defined as follows:
\begin{center}
    \begin{tabular}{{r c l}}
        !prod(A:$\Array{\reals}{n}$)! &:=& !reduce * 1 A! \\
        !sum(A:$\Array{\reals}{n}$)! &:=& !reduce + 0 A! \\
        !dot(A:$\Array{\reals}{n}$,B:$\Array{\reals}{n}$)! &:=& !reduce + 0 (map2 * A B)!     
    \end{tabular}
\end{center}
where !reduce! is a known fold-left operator for which the function argument is associative. 
It is notably faster to execute than a fold-left, as it is parallel-friendly.

The gradient of each of these expressions with respect to !A! is:
\begin{center}
    \begin{tabular}{{r c l}}
        $\nabla_A$!prod(A)! &:=& !map2 * (scanr * 1 (shift1L A)) (shift1R (scanl * 1 A))! \\
        $\nabla_A$!sum(A)! &:=& !map (x -> 1) A!\\
        $\nabla_A$!dot(A,B)! &:=& !B! 
    \end{tabular}
\end{center}
where
\begin{itemize}
\item !scanl! is the scan-left operator that returns all the intermediate results of fold-left,
\item !scanr! is the scan-right operator that returns all the intermediate results of fold-right,
\item !shift1L [v$_1$,$\ldots$,v$_n$]! is the shift-left operator and returns ![v$_2$,$\ldots$,v$_{n}$]!, and 
\item !shift1R [v$_1$,$\ldots$,v$_n$]! is the shift-right operator and returns ![v$_1$,$\ldots$,v$_{n-1}$]!.
\end{itemize}
These gradients are a few examples among numerous ones which are usually derived by hand, and are here obtained automatically as special cases of our work.
\end{example}   

\begin{figure}
\[
\begin{tikzcd}
    \Source \ar[rrrr,"\text{efficient }\Dsynrevsymbol (Fig.~\ref{fig:direct_diff_macro})"] \ar[dr,"(Fig.~\ref{fig:source_to_unf})"'] &&&& \Target \arrow[r,loop right,"\text{optim }(Fig.~\ref{fig:optim})"] \\
    & Source\UNFSymbol \ar[rr,"\Dsynrevsymbol (Fig.~\ref{fig:diff_macro})"'] && Target\UNFSymbol \ar[ur,"(Fig.~\ref{fig:unf_to_target})"'] & 
\end{tikzcd}
\]
\vspace{-0.5cm}
\caption{Outline of the compilation scheme.}
\vspace{-0.5cm}
\label{fig:outline}
\end{figure}

\subsection{Contributions}
We propose a source-code transformation on a simple purely functional language for purely functional reverse-mode differentiation.
Our transformation consists of a compilation scheme that is outlined in Figure~\ref{fig:outline}.
We make the following contributions:

\begin{itemize}[leftmargin=*]
\item We present our work with a simple yet expressive array-based language (with constructs such as !map2! and !reduce!) in Section~\ref{sec:simplediff}. 
We show how to directly compute an efficient reverse-mode AD for the expressions of this program (top of Figure~\ref{fig:outline}).  
Furthermore, we show how to extend our work to a richer language in Section~\ref{sec:generalization}.
\item One of the key insights behind efficient reverse-mode AD is to only consider unary operators. 
Inspired by this insight and following Intermediate Representations (IR) such as SSA and ANF, we introduce a novel IR, which we call UNF (Section~\ref{sec:unf}).
We introduce an alternative and easier-to-follow compilation pipeline for efficient reverse-mode AD (bottom of Figure~\ref{fig:outline}).
\item We prove complexity guarantees for the programs transformed under reverse-mode AD. Furthermore, we show a list of optimizations that can further improve the constant factors (Section~\ref{sec:complexity}).
\item We prove the correctness of our transformations (top/bottom parts of Figure~\ref{fig:outline}) by defining a denotational semantics of our languages using multicategories and concategories (Section~\ref{sec:correctness}).
\end{itemize}

Next, we recall rudiments of automatic differentiation, forward and reverse-mode differentiation. 
\section{Reverse-mode Automatic Differentiation}
\label{sec:background}

\subsection{Rudiments of AD and dual numbers}

To find the gradient $\nabla f$ of a function $f:\RR\to \RR$ compositionally, need both $f$ and $\nabla f$ when calculating $\nabla (f;g)$.
This is the reason why we are more generally interested in transforming a function $f:\RR^n\to \RR$ into a function
$g:(\RR\times \RR)^n\to \RR\times \RR$ in such a way that for every
$f_1, \dots, f_n:\RR\to\RR$, 
\begin{center}
    $(f_1,\nabla f_1,\dots, f_n,\nabla f_n);g = ((f_1,\dots, f_n);f,\nabla ((f_1, \dots, f_n);f))$.
\end{center}

The idea of AD is to systematically transform a differentiable function $f:\RR^n\to \RR$ into a function $g:\RR^{2n}\to \RR^2$ which captures~$f$ and all its partial derivatives.
An intuition for $g$ is often given in terms of dual numbers. The transformed function operates on pairs of numbers, $(x,x')$, and it is common
to think of such a pair as $x+x'\epsilon$ for an `infinitesimal' $\epsilon$. The main two ways in which AD is performed in practice is by operator overloading or by source code transformation (see e.g. \cite{griewank2008evaluating} Chapter 6). 
Our approach focuses on a source code transformation, which is better fitted for compilation and optimizations.

\subsection{Reverse-mode Automatic Differentiation}

A potential computational problem shows up when one wants to compute the full gradient of a function $\RR^n\to\RR$, for a large $n$. 
Forward-mode only computes one directional derivative, for instance one partial derivative. 
This implies $n$ passes must be performed through the forward derivative to compute the whole gradient.
By using the symmetry in the chain rule, there is a way to compute the whole gradient faster, and this method is reverse-mode automatic differentiation.
Suppose given a function $f=f_n\circ\ldots\circ f_1 :\RR^n\to\RR$. 
Mathematically, forward mode essentially computes $(\J f)v=\J f_n(\J f_{n-1}(\ldots(\J f_1v))\ldots)$ for a direction $v\in\RR^n$. 
Reverse-mode, on the other hand, computes $(\J f)^Tv = \J^Tf_1(\J^T f_{2}(\ldots(\J^T f_nv))\ldots)$ for a vector $v\in\RR$.
In particular, taking $v=1$ computes the gradient of $f$.

Because the computation flow of the function is reversed, the actual implementation of reverse-mode is quite tricky. 
Reverse-mode AD is only well-understood as a source-code transformation on limited programming languages. 
Typically, its implementations on more expressive languages that have features such as higher-order functions and conditionals
make use of define-by-run approaches. These approaches first build a computation graph during runtime, effectively evaluating the program until a straight-line
first-order program is left, and then they evaluate this new program \cite{carpenter2015stan,paszke2017automatic}. 
Such approaches have the severe drawback that the obtained code cannot benefit from existing optimizing compilers.
As such, the implementation process is tedious and labor-intensive as these AD libraries need to be implemented using carefully, manually optimized code.
In addition, some whole-program optimizations that a compiler would detect are completely missed.

\subsection{Inefficiency of purely functional reverse-mode AD}

Following \cite{pearlmutter2008reverse}, there is a simple way to define an inefficient yet purely functional reverse-mode transformation for first-order programs.
We review a slight modification of their transformation, which is also better explained through an example. 

Let us consider the term $x_1:\RR,\ldots,x_n:\RR\vdash \exp(\cos(x_i))$.
To compute its gradient, following the chain rule, we need the Jacobian matrices of $\cos$ at $x_i$ and of $\exp$ at $\cos(x_i)$. 
Instead of considering these as operations from $\RR\to\RR$, we consider them as functions from the whole context. So $\cos$ and $\exp$ are seen as functions $\RR^n\to\RR$.
However, by simply doing this, we lose compositionality. 
So we modify $\cos$ to also return its context. It is now seen as a function $\sem{\cos}:\RR^n\to\RR^{n+1}$.
Similarly, $\exp$ is transformed. It also needs to take the return value of $\sem{\cos}$ as an extra argument, the one it will actually use and not simply return. 
We thus obtain $\sem{\exp}:\RR^{n+1}\to\RR^{n+2}$. Now the jacobians matrices $\J\sem{\cos} \in Mat_{n,n+1}$, $\J\sem{\exp} \in Mat_{n+1,n+2}$ compose nicely.
The same can be done for binary operators and let bindings.
This transforms a first-order program to a function $f:\RR^n\to\RR^{n+m}$ of the form $f_m\circ\ldots\circ f_1$. 
If the original program was of type $\RR$, then the return value of the original program is the last component of $f$.
Following the mathematical presentation of reverse-mode above, the gradient of the original program is then obtained as 
\begin{center}
    \begin{tabular}{r c l}
        $\nabla f$ &=& $\J^Tf(0,\ldots,0,1)= \J^Tf_1(J^Tf_{2}(\ldots(\J^Tf_m(0,\ldots,0,1))\ldots)$
    \end{tabular}
\end{center}

To actually reverse the order of computation needed for this transpose of Jacobians, we use a simple continuation;
$f_i$ is turned into $\Dsynrevsymbol{f_i}:=<f_i, \lambda Y. Y\circ J^Tf_i>$ where $Y:\RR^{n+i-1}\to \RR^n$. 
We recover compositionality by noting that $<f_{i+1}(f_i), (\lambda Y. Y\circ J^Tf_{i+1})(\lambda Y. Y\circ J^Tf_i)>$ reduces to
$<f_{i+1}\circ f_i, \lambda Y. Y\circ J^Tf_i \circ J^Tf_{i+1}>$, and thus by induction we can obtain $<f, \lambda Y. YJ^Tf>$.
By applying the identity continuation $\RR^n\to\RR^n$ on the second component and then the result to $(0,\ldots,0,1)$, 
we have obtained a purely functional way to compute $\nabla f$. 

This purely functional implementation has the following issues in terms of efficiency: 

\smartpara{Issue 1}
If we see the term as a directed graph, reverse mode back propagates from the end of the graph to the starting nodes via every path.
However, it is hard to keep track of all this information in parallel in a functional way.
Mutation is usually key for these cases; 
the imperative version of reverse mode for a binary operator $op(x,y)$ adds $x'\pluseq \partial_1op(x,y);y'\pluseq\partial_2op(x,y)$, 
where $\partial_iop$ are the partial derivatives of $op$.

\smartpara{Issue 2} 
Each $J^Tf_{i+1}$ is a potentially huge matrix if $n$ or $m$ is big.

\smartpara{Issue 3}
We have to carry a continuation and $\beta$-reduce a lot of higher-order functions.

\subsection{Insights for efficient purely functional reverse-mode AD}
\label{subsec:insights}
Overall, we use the following three insights to solve the inefficiency associated with the purely functional implementations of reverse-mode AD.

\smartpara{Insight 1}
One of the key simple ideas that we used was to transform every operator into a unary one. 
This essentially trivializes the computation flow to a line. 
Even if the starting program was a straight-line program, 
having non-unary operators was a source of inefficiency and justified the use of mutation in the first place.
By returning every variable every time, the problem of using a variable several times does not need to be dealt with via mutation. 
This simple idea of transforming a program into essentially a straight line is what our new intermediate representation UNF allows.

\smartpara{Insight 2} 
If we look at $J^Tf_{i+1}$, we notice that this function is almost the identity, except at the last row. 
Even on the last row, if the original term was a unary or binary operator like $cos, exp, +, *$, 
the row is zero except for at most two indices (one for unary operators).
This means we can use a more compact representation $J^T\sem{op(x_i,x_j)}:=\lambda (y_1,\ldots,y_{n+i}).(y_1,\ldots,y_{n+i})+[i]\partial_1op(x_i,x_j)+[j]\partial_2op(x_i,x_j)$, 
where $[k]$ means that the element is added at the $k$-th index of the tuple.

\smartpara{Insight 3}
We know in advance that the Jacobian functions are going to be applied one to another, and we can use partial evaluation to $\beta$-reduce all of these $\lambda$s.
Because each function is almost the identity, we obtain a lot of substitutions of the form $[x/y]$ where both $x$ and $y$ are variables. 
This allows us to drastically reduce the size of $\J^Tf$. In fact, for simple programs, this is basically enough to get an efficient purely functional reverse derivative transformation.
We develop this idea further for a richer language.
 
\section{Simple pure reverse-mode differentiation}
\label{sec:simplediff}

\subsection{Source Language}
\label{sub:sourcelang}

We consider a standard language, and give it a standard call-by-value operational semantics. 
It consists of a first-order functional language with arrays and a few typical second-order array operations. 
The types !T!, terms !e!, and typing rules are given in Figure~\ref{fig:source_grammar}.
We have included a minimal set of array operations for the sake of illustration,  it is not hard to add more.
See Section~\ref{sec:generalization}.

\begin{figure*}[t]
\setlength{\tabcolsep}{0.3em}
\centering
\begin{tabular}{|l c l|l|}
\hline
\multicolumn{3}{|c|}{\textbf{Core Grammar}} & \multicolumn{1}{c|}{\textbf{Description}}\\\hline
!T! & \mbox{::=} & $\reals$ & \grammarcomment{Real Type} \\
& $\mid$ & !T! $\times$ !T! & \grammarcomment{Product Type}\\
& $\mid$ & $\Array{\reals}{n}$ & \grammarcomment{Real Array Type of size $n$}\\
\hline
!e! & \mbox{::=} & !x! $\mid$ !c! & \grammarcomment{Variable \& Real constant}\\
& $\mid$ & !let x = e in e! & \grammarcomment{Variable Binding}\\
& $\mid$ & !< e, e >! $\mid$ $\pi_1$(!e) $\mid$ $\pi_2$(!e) & \grammarcomment{Pair Constructor/Destructor}\\
& $\mid$ & !e op2 e! $\mid$ !op1 e! & \grammarcomment{Binary/Unary operations}\\
& $\mid$ & !map2 (x,y.e) e e! $\mid$ !reduce (x,y.e) e e! & \grammarcomment{Array map2 \& reduce }\\
\hline
\end{tabular} \\ \vspace{0.2cm}
\begin{tabular}{|c|}
    \hline
    \begin{tabular}{c} 
    \\\hline
    $\Gamma \vdash$ !x!: !T!
    \end{tabular}~(!x!: !T!$\in\Gamma$)
    \hspace{0.5cm}
    \begin{tabular}{c}
    $\Gamma \vdash$ !e$_1$!: !T$_1$! $\quad$ $\Gamma \vdash$ !e$_2$!: !T$_2$! \\\hline  
    $\Gamma \vdash$ !<e$_1$,e$_2$>!: !T$_1$! $\times$ !T$_2$!
    \end{tabular}
    \hspace{0.5cm}
    \begin{tabular}{c}
        $\Gamma \vdash$ !e!: !T$_1$! $\times$ !T$_2$! \\\hline  
        $\Gamma \vdash$ $\pi_i$!e!: !T$_i$!
    \end{tabular}~($i\in\{1,2\}$)
\\
    \begin{tabular}{c}
    $\Gamma \vdash$ !e$_1$!: !T$_1$! $\quad$ $\Gamma$, !x!: !T$_1$! $\vdash$ !e$_2$!: !T$_2$! \\\hline
    !let x = e$_1$ in e$_2$!: !T$_2$!
    \end{tabular}
    \hspace{0.5cm}
    \begin{tabular}{c}
        $\Gamma \vdash$ !e!: $\reals$ \\\hline  
        $\Gamma \vdash$ !op1 e! : $\reals$
    \end{tabular}
    \hspace{0.5cm}
    \begin{tabular}{c}
        $\Gamma \vdash$ !e$_1$!: $\reals$ $\quad$ $\Gamma \vdash$ !e$_2$!: $\reals$ \\\hline  
        $\Gamma \vdash$ !e$_1$ op2 e$_2$! : $\reals$
        \end{tabular}
\\
    \begin{tabular}{c}
        $\Gamma$, !x!: $\reals$, !y!: $\reals$ $\vdash$ !e$_1$!: $\reals$ 
        $\quad$ $\Gamma$ $\vdash$ !e$_2$!: $\Array{\reals}{n}$
        $\quad$ $\Gamma$ $\vdash$ !e$_3$!: $\Array{\reals}{n}$
        \\\hline  
        $\Gamma \vdash$ !map2 (x,y.e$_1$) e$_2$ e$_3$!: $\Array{\reals}{n}$
    \end{tabular}
\\
    \begin{tabular}{c} 
        \\\hline
        $\Gamma \vdash$ !c!: $\reals$
    \end{tabular}
    \hspace{0.5cm}
    \begin{tabular}{c}
        !x!: $\reals$, !y!: $\reals$ $\vdash$ !e$_1$!: $\reals$ 
        $\quad$ $\Gamma$ $\vdash$ !e$_2$!: $\reals$
        $\quad$ $\Gamma$ $\vdash$ !e$_3$!: $\Array{\reals}{n}$
        \\\hline  
        $\Gamma \vdash$ !reduce (x,y.e$_1$) e$_2$ e$_3$!: $\reals$
    \end{tabular} \\ \hline
\end{tabular}
\vspace{-0.4cm}
\caption{Grammar and type system of the source language.}
\vspace{-0.4cm}
\label{fig:source_grammar}
\end{figure*}

For scalar operations, we assume a given set of operations, including $+$ and $*$. 
!op1! and !op2! denote respectively a unary and a binary operation on reals. 
These operations represent smooth total functions, 
but again, this can be easily generalized (\S\ref{sec:generalization}).  
Typical examples include !cos, exp, +, *!. 
We use infix notation for binary operators.

The !reduce! operator is a fold-left operator for which the function is assumed to be associative, 
and the provided initial value should be a unit of the binary operation.
It is a well-known parallel-friendly construct. 
For the sake of simplicity in the presentation, the bound function in !reduce! 
is restricted to having no free variables. Furthermore, as our main focus is on AD and not array processing,
we currently restrict to arrays of reals.
We show how to lift these restrictions and how to differentiate some other array operators in Section~\ref{sec:generalization} 
and in the supplementary material.

\subsection{Target Language}

The target language of our source-code transformation is an extension to the source language.
It is a higher-order language, as our purely functional reverse-mode introduces a continuation.
The set of scalar operations should also be closed under partial differentiation. 
In more detail, for every unary scalar operation !op1!, 
we assumed a given operator $\partial$!op1! whose semantics should be the derivative of !op1!, e.g. $\partial$!sin!=!cos!.
Similarly, for every binary operator !op2!, we assume given operators $\partial_1$!op2!, $\partial_2$!op2!, 
respectively representing the first and second partial derivative of !op2!.

Similarly, the target language contains more array primitives, which are used to define the reverse derivatives of array operations. 
Scan left !scanl! is similar to fold left, but also stores all the intermediate results in an array, which it returns.
In the same vein, scan right !scanr! performs a fold left by reading the array from right to left and stores 
the intermediate results in an array from right to left.

Finally, we add two new shift operators !shift1L! and !shift1R!. 
They take an array of size $n$, and respectively forget the first and the last element of the array.
These somewhat ad-hoc operators naturally show up when differentiating fold-like operators.

The grammar for types and terms along with the type system of the target language are presented in Figure~\ref{fig:target_grammar}.
Our lambda abstractions take $n$ arguments, as we are not concerned with partial applications in this work. 
In fact, the lambda abstractions introduced by reverse-mode will be removed during partial evaluation, 
and the notation with lambda abstractions having $n$ bound variables makes reading slightly easier.
We note that we don't actually need the full power of higher-order because we only use lambda abstractions over variables of ground types
and let expressions binding such lambda abstractions. We only need the target language to be second-order.

\begin{figure*}[t]
    \setlength{\tabcolsep}{0.3em}
    \centering
    \begin{tabular}{|l c l|l|}
    \hline
    \multicolumn{3}{|c|}{\textbf{Core Grammar}} & \multicolumn{1}{c|}{\textbf{Description}}\\\hline
    !T! & \mbox{::=} & $\ldots$ & \grammarcomment{Same as Source} \\
    & $\mid$ & !T!$\times \ldots \times$!T! & \grammarcomment{$n$-tuples}\\ 
    & $\mid$ & !T->T! & \grammarcomment{Function Type}\\ 
    \hline
    !e! & \mbox{::=} & $\ldots$ & \grammarcomment{Same as Source}\\
    & $\mid$ & !fun (x$_1$,$\ldots$,x$_n$) -> e! & \grammarcomment{Lambda Abstraction}\\
    & $\mid$ & !e(e$_1\ldots$e$_n$)! & \grammarcomment{Function Application}\\
    & $\mid$ & !<e,$\ldots$,e>! & \grammarcomment{Tuples}\\
    & $\mid$ & !scanl (x,y.e) e e! $\mid$ !scanr (x,y.e) e e! & \grammarcomment{Array scan left and right}\\
    & $\mid$ & !shift1L e! $\mid$ !shift1R e! & \grammarcomment{Array left/right shifting}\\
    \hline
    \end{tabular}\\ \vspace{0.2cm}
    \begin{tabular}{|c|}
    \hline
    \begin{tabular}{c}
        $\Gamma$, !x!: $\reals$, !y!: $\reals$ $\vdash$ !e$_1$!: $\reals$ 
        $\quad$ $\Gamma$ $\vdash$ !e$_2$!: $\reals$
        $\quad$ $\Gamma$ $\vdash$ !e$_3$!: $\Array{\reals}{n}$
        \\\hline  
        $\Gamma \vdash$ !scanl (x,y.e$_1$) e$_2$ e$_3$!: $\Array{\reals}{n+1}$
    \end{tabular}
\\
    \begin{tabular}{c}
        $\Gamma$, !x!: $\reals$, !y!: $\reals$ $\vdash$ !e$_1$!: $\reals$ 
        $\quad$ $\Gamma$ $\vdash$ !e$_2$!: $\reals$
        $\quad$ $\Gamma$ $\vdash$ !e$_3$!: $\Array{\reals}{n}$
        \\\hline  
        $\Gamma \vdash$ !scanr (x,y.e$_1$) e$_2$ e$_3$!: $\Array{\reals}{n+1}$
    \end{tabular}
\\
    \begin{tabular}{c}
        $\Gamma$, !x$_1$!: !G$_1$!, $\ldots$, !xn!: !G$_n$! $\vdash$ !e!: !T! 
        \\\hline  
        $\Gamma \vdash$ !fun (x$_1$,$\ldots$,x$_n$) -> e!: !G$_1\times\ldots\times$G$_n$->T!
    \end{tabular}
\\
    \begin{tabular}{c}
        $\Gamma$ $\vdash$ !e!: !G$_1\times\ldots\times$G$_n$ -> T!
        $\quad$ $\Gamma$ $\vdash$ !e$_i$!: !Gi! for all $1\leq i\leq n$
        \\\hline  
        $\Gamma \vdash$ !e(e$_1\ldots$e$_n$)!: !T!
    \end{tabular}
\\
    \begin{tabular}{c}
        $\Gamma$ $\vdash$ !e!: $\Array{\reals}{n+1}$
        \\\hline  
        $\Gamma \vdash$ !shift1L e!: $\Array{\reals}{n}$
    \end{tabular}
    \hspace{0.5cm}
    \begin{tabular}{c}
        $\Gamma$ $\vdash$ !e!: $\Array{\reals}{0}$
        \\\hline  
        $\Gamma \vdash$ !shift1L e!: $\Array{\reals}{0}$
    \end{tabular}
\\
    \begin{tabular}{c}
        $\Gamma$ $\vdash$ !e!: $\Array{\reals}{n+1}$
        \\\hline  
        $\Gamma \vdash$ !shift1R e!: $\Array{\reals}{n}$
    \end{tabular}
    \hspace{0.5cm}
    \begin{tabular}{c}
        $\Gamma$ $\vdash$ !e!: $\Array{\reals}{0}$
        \\\hline  
        $\Gamma \vdash$ !shift1R e!: $\Array{\reals}{0}$
    \end{tabular} 
    \\ 
    \begin{tabular}{c}
        for all $i$, $\Gamma$ $\vdash$ !e$_i$!: !T$_i$!
        \\\hline  
        $\Gamma \vdash$ !<e$_1$,$\ldots$,e$_n$>!: !T$_1\times\ldots\times$T$_n$!
    \end{tabular} \\
    \hline
    \end{tabular}
    \vspace{-0.4cm}
    \caption{Grammar and type system of the target language.}
    \label{fig:target_grammar}
    \vspace{-0.4cm}
    \end{figure*}

\subsection{Macro for pure reverse mode transformation} 
\label{sub:Macro for pure reverse mode transformation}

In Figure~\ref{fig:direct_diff_macro} we present our direct transformation from the source language to the target language for pure reverse mode differentiation.
Given a term $\Gamma\vdash e : \reals$, we can compute its gradient $\grad_\Gamma e$ from a particular instance of 
$\directD{\rho}{\Gamma}{Y}(e)$. First, $\rho, Y$ specifies if we want to compute the whole gradient regarding the variables from $\Gamma$ or a subset of it.
For a subset $\rho\subset \Gamma$, one chooses $Y$ to be the projection function sending a variable 
$x_i:G$ of $\Gamma$ to $x_i$ if it belongs to $\rho$ and to $0_G$ otherwise.
In particular, we take $Y=Id_\Gamma$ to compute the whole gradient.
Next, the gradient will be given by the second part of the pair $\directD{\rho}{\Gamma}{Y}(e)$, 
and we need to initialize the tangent variables. All of them are set to $0$, except the one corresponding to the output value of !e!, 
which we initialize at $1$ to run the backpropagation. 
All in all, we compute the gradient via $\pi_2\directD{\rho}{\Gamma}{Id_\Gamma}(e)(0_\Gamma,1)$.

\begin{figure}
\begin{tabular}{|l c l|}
    \hline
    !let x$_{1}$=e$_{1}$,$\ldots$,x$_n$=e$_n$!  & \multirow{2}{*}{=} & !let x$_{1}$=e$_{1}$ in let x$_{2}$=e$_{2}$ in $\ldots$! \\
    !in e! && !let x$_n$ = e$_n$ in e!\\ \hline
    $Id_\Gamma$ \quad\quad\quad ($\Gamma \, = \, x_1:A_1,\ldots,x_n:A_n)$ & = & !fun! $(y_1:A_1,\ldots,y_n:A_n)$! -> !$(y_1,\ldots,y_n)$ \\ \hline
    $\grad_\Gamma (e)$ \quad\quad\hspace{0.6em}($\Gamma\vdash e:\RR$) & = & $\pi_2\directD{\Gamma}{\Gamma}{Id_\Gamma}(e)(0_\Gamma,\underline{1})$ \\ \hline
    !pos(x)! \quad(!x!$\in\Gamma=x_:A_1,\ldots,x_n:A_n)$ & = & position $i$ of !x! in $\Gamma$ \\ \hline
    ![i]e! \quad\quad(!e! of ground type $G_i$) & \multirow{2}{*}{=} &  $(0_{G_1},\ldots,0_{G_{i-1}},e,0_{G_{i+1}},\ldots,0_{G_n})$ \\
    && ($G_j$ are ground types) \\ \hline
    $\grad_{\Gamma_1}$(!e!) \quad\quad($\Gamma=\Gamma_1,\Gamma_2$, $|\Gamma_1|=k$)& = & $(e_1,\ldots,e_k)$ \\
    $\grad_{\Gamma_2}$(!e!) & = & $(e_{k+1},\ldots,e_n)$ \\
    && when $\grad_{\Gamma}($!e!$) = (e_1,\ldots,e_n)$ \\ \hline
    !map (x. e) A! & = & !map2 (x y. e) A A! \\ \hline
    !ZerosLike(A)! & = & !map (x.0) A! \\ \hline
    !OnesLike(A)! & = & !map (x.1) A! \\ \hline
    !ZerosLike(n)! & = & !map (x.0) (\_:$\Array{\reals}{n}$)! \\ \hline
    !OnesLike(n)! & = & !map (x.1) (\_:$\Array{\reals}{n}$)! \\ \hline
\end{tabular}
\vspace{-0.4cm}
\caption{Notations used for reverse-mode AD transformation. $\_$ represents a dummy variable added to the context.
}
\vspace{-0.4cm}
\label{tbl:notation:one}
\end{figure}

\begin{notation}
We now introduce several notations which are useful when defining the transformation for reverse-mode in Figure~\ref{tbl:notation:one}.
Ground types are defined inductively by
\begin{center}
    $G::= \RR \mid G\times \ldots \times G \mid \Array{\RR}{n}$
\end{center} 

$(\RR,+,\underline{0})$ forms a monoid, and this monoid structure extends canonically 
to a monoid structure $(G,\widehat{+},0_G)$ for every ground type $G$. 
It is defined inductively on $G$ as follows

\begin{tabular}{l c l}
    $0_\RR$  & $\defeq$ & $\underline{0}$ \\
    $0_{G_1\times \ldots \times G_n}$ & $\defeq$ &  $< O_{G_1},\ldots, 0_{G_n} >$ \\
    $0_{\Array{\RR}{n}}$& $\defeq$ & !ZerosLike(n)!
\end{tabular}
\begin{tabular}{l c l}
    $a\widehat{+}_\RR b$ & $\defeq$ & $ a+b$ \\
    $(a_1,\ldots,a_n)\widehat{+}_{G_1\times\ldots\times G_n}$ & $\defeq$ & $(a_1\widehat{+}_{G_1}b_1,\ldots,a_n\widehat{+}_{G_n}b_n)$ \\
    $(b_1,\ldots,b_n)$ && \\
    $A\widehat{+}_{\Array{\RR}{n}}B $ & $\defeq$ & !map2! + $A$ $B$ 
\end{tabular}

A ground context is a context only containing variables of ground type.
The previous monoid structure again extends canonically to ground contexts $\Gamma$ by defining
$0_{x1:G1,\ldots,x_n:G_n}\defeq 0_{G_1},\ldots,0_{G_n}$ and 
$a_1,\ldots,a_n\widehat{+}_{x1:G_1,\ldots, x_n:G_n}b_1,\ldots,b_n\defeq a_1\widehat{+}_{G_1}b_1,\ldots,a_n\widehat{+}_{G_n}b_n$.

\end{notation}

\begin{example}
    The reverse-mode transformation of the terms from the introduction are given by

    \begin{tabular}{c l}
        &$\directD{\rho}{\Gamma}{Y}$(!let w$_1$ = x$_1$ * x$_2$ in let w$_2$ = w$_1$ * x$_1$ in w$_2$!) \\
        =& !let w$_1$,Y$_1$=!\\
        & \quad\quad !let y$_{11}$,Y$_{11}$= <x$_1$, fun (y$_1$,y$_2$,y$_3$,z) -> Y(y$_1$+z,y$_2$,y$_3$)> in! \\
        & \quad\quad !let y$_{12}$,Y$_{12}$= <x$_{2}$, fun (y$_{1}$,y$_{2}$,y$_{3}$,y$_{4}$,z) -> Y$_{11}$(y$_{1}$,y$_{2}$+z,y$_{3}$,y$_{4}$)> in! \\
        & \quad\quad !<y$_{11}$ * y$_{12}$, fun (y$_{1}$,y$_{2}$,y$_{3}$,z) -> Y$_{12}$(y$_{1}$,y$_{2}$,y$_{3}$,y$_{12}$*z,y$_{11}$*z) > in! \\
        & !let w$_{2}$,Y$_{2}$=!\\
        & \quad\quad !let y$_{21}$,Y$_{21}$= <w$_{1}$, fun (y$_{1}$,y$_{2}$,y$_{3}$,y$_{4}$,z) -> Y$_1$(y$_{1}$,y$_{2}$,y$_{3}$,y$_{4}$+z)> in! \\
        & \quad\quad !let y$_{22}$,Y$_{22}$= <x$_{1}$, fun (y$_{1}$,y$_{2}$,y$_{3}$,y$_{4}$,y$_{5}$,z) -> Y$_{21}$(y$_{1}$+z,y$_{2}$,y$_{3}$,y$_{4}$,y$_{5}$)> in! \\
        & \quad\quad !<y$_{21}$ * y$_{22}$, fun (y$_{1}$,y$_{2}$,y$_{3}$,y$_{4}$,z) -> Y$_{22}$(y$_{1}$,y$_{2}$,y$_{3}$,y$_{4}$,y$_{22}$*z,y$_{21}$*z) > in! \\
        & !let y,Y$_{3}$= <w, fun (y$_{1}$,y$_{2}$,y$_{3}$,y$_{4}$,z) -> Y$_{2}$(y$_{1}$,y$_{2}$,y$_{3}$,y$_{4}$+z)> in! \\
        & !<y, fun (y$_{1}$,y$_{2}$,y$_{3}$,z) -> Y$_{3}$(y$_{1}$,y$_{2}$,y$_{3}$,0,z) >!
    \end{tabular}
    \medskip

    \begin{tabular}{c l}
        &$\directD{\rho}{\Gamma}{Y}$(!prod(A)!) \\
        =& !let y,Y$_{1}$= <1, fun (X,z) -> Y(X)> in! \\
        & !let B,Y$_{2}$= <A, fun (X,x,Z) -> Y$_{1}$(X+Z,x)> in!\\
        & !let A$_{0}$= shift1R (scanl * y B) in!\\
        & !let A$_{1}$= shift1L (map2 (a,b.b) A$_{0}$ B) in!\\
        & !let A$_{2}$= map2 (a,b.a) A$_{0}$ B in!\\
        & !let A$_{3}$= scanr * 1 A$_{1}$ in!\\
        & !<prod(B), fun (X,z) -> Y$_{2}$(X,0,map2 (a,b. a*z*b) A$_{2}$ A$_{3}$)>! 
    \end{tabular}
\end{example}

The idea is that $\rho$ represents the return type of the derivative part, which should be !A$_{1} \times \ldots \times$ A$_n$! 
if we want the whole gradient of a term !e! in context !$\Gamma = \;$x$_{1}$:A$_{1}$,$\ldots$,x$_n$:A$_n$!.
The subscript $\Gamma$ denotes the current context, 
which is locally augmented, for instance when differentiating a !let! rule.
For non-unary operations, we differentiate the arguments from the left to the right and add their derivatives to the current stack, 
which is modeled by the continuation $Y$. 
Importantly for performance, each continuation variable $Y$ is only used once.

We have the following typing lemma for $\directD{\rho}{\Gamma}{Y}$, routinely proved by induction on derivation of $\Gamma\vdash$!e:A!.
\begin{lemma}[Typing $\directD{\rho}{\Gamma}{Y}$]
    If $\Gamma \vdash$ !e: A!, then $\Gamma$!,Y:!$\Gamma\to \rho \vdash \directD{\rho}{\Gamma}{Y}$!(e): !$\directD{\rho}{\Gamma}{Y}$!(A)!.
\end{lemma}

\begin{figure*}[t]
    \setlength\tabcolsep{1.5pt}
    \small
    \begin{tabular}{|r c l|}
    \hline
        $\ad{\Gamma}{Y}$(!A!) &=& !A $\times$ ($\Gamma$ $\times$ A->$\rho$)!\\ & & \\
        $\ad{\Gamma}{Y}$(!c!) &=& 
            !<c, fun ($\boldsymbol{x}$,z) ->! \\
            && !Y($\boldsymbol{x}$) >!\\
        $\ad{\Gamma}{Y}$(!x!) &=& 
            !<x, fun ($\boldsymbol{x}$,z) ->! \\
            && !Y($\boldsymbol{x}$!$\widehat{+}$![pos(x)]z) >!\\
        $\ad{\Gamma}{Y}$(!let x:A = e$_{1}$ in e$_{2}$!) &=& 
            !let x,Y$_{1}$ = !$\ad{\Gamma}{Y}$!(e$_{1}$) in! \\
            &&!let y,Y$_{2}$ = !$\ad{\Gamma,x:A}{Y_1}$!(e$_{2}$) in!\\ 
            &&!<y, fun ($\boldsymbol{x}$,z) -> Y$_{2}$($\boldsymbol{x}$,!$0_{A}$!,z)>!\\
        $\ad{\Gamma}{Y}$(!< e$_{1}$, e$_{2}$ >!) &=&
            !let y$_{1}$,Y$_{1}$ = !$\ad{\Gamma}{Y}$!(e$_{1}$) in! \\
            &&!let y$_{2}$,Y$_{2}$ = !$\ad{\Gamma,x_1}{Y_1}$!(e$_{2}$) in!\\
            &&!< <y$_{1}$,y$_{2}$>, fun ($\boldsymbol{x}$,z) -> !\\
            &&!Y($\boldsymbol{x}$,!$\pi_1$!(z),!$\pi_2$!(z)) >!\\ 
        $\ad{\Gamma}{Y}$($\pi_1$(!e!:A$\times$B)) &=&
            !let x,Y$_{1}$ = !$\ad{\Gamma}{Y}$!(e) in! \\
            && !<!$\pi_1$!x, fun ($\boldsymbol{x}$,z) -> Y($\boldsymbol{x}$,(z,!$0_B$!))>! \\
        $\ad{\Gamma}{Y}$($\pi_2$(!e!:A$\times$B)) &=&
            !let x,Y$_{1}$ = !$\ad{\Gamma}{Y}$!(e) in! \\
            && !<!$\pi_2$!x, fun ($\boldsymbol{x}$,z) -> Y($\boldsymbol{x}$,(!$0_A$,z!))>! \\
        $\ad{\Gamma}{Y}$(!op1 e!) &=&  
            !let x,Y$_{1}$ = !$\ad{\Gamma}{Y}$!(e) in! \\
            && !<op1 x, fun ($\boldsymbol{x}$,z) -> ! \\
            && !Y($\boldsymbol{x}$,!$\partial$!op1(x)*z) >! \\
        $\ad{\Gamma}{Y}$(!e$_{1}$ op2 e$_{2}$!) &=& 
            !let x$_{1}$,Y$_{1}$ = !$\ad{\Gamma}{Y}$!(e$_{1}$) in! \\
            && !let x$_{2}$,Y$_{2}$ = !$\ad{\Gamma,x_1}{Y_1}$!(e$_{2}$) in! \\
            && !<x$_{1}$ op2 x$_{2}$, fun ($\boldsymbol{x}$,z) ->! \\
            && !Y$_{2}$($\boldsymbol{x}$,!$\partial_1$!op2(x$_{1}$,x$_{2}$)*z,!$\partial_2$!op2(x$_{1}$,x$_{2}$)*z>! \\
        $\ad{\Gamma}{Y}$(!map2 (x,y.e$_{1}$) e$_{2}$ e$_{3}$!) &=&  
            !let A,Y$_{1}$ = !$\ad{\Gamma}{Y}$!(e$_{2}$) in! \\
            && !let B,Y$_{2}$ = !$\ad{\Gamma,A}{Y_1}$!(e$_{3}$) in! \\
            && !<map2 (x,y.e$_{1}$) A B, fun ($\boldsymbol{x}$,Z) -> !\\
            && !let G = (map2 * Z!\\
            && !  (map2 (a,b.(!$\grad_\Gamma$!e$_{1}$)[a/x.b/y]) A B)) in! \\
            && !Y$_{2}$( $\boldsymbol{x}$ $\widehat{+}$(reduce $\widehat{+}$ $\widehat{0}$ G),!\\
            && !map2 * (map2 (a,b.(!$\grad_{\{x\}}$!e$_{1}$)[a/x,b/y]) A B) Z,!\\
            && !map2 * (map2 (a,b.(!$\grad_{\{y\}}$!e$_{1}$)[a/x,b/y]) A B) Z)>!\\
        $\ad{\Gamma}{Y}$(!reduce (x,y.e$_{1}$) e$_{2}$ e$_{3}$!) &=&
            !let y$_{1}$,Y$_{1}$ = !$\ad{\Gamma}{Y}$!(e$_{2}$) in! \\
            && !let A,Y$_{2}$ = !$\ad{\Gamma,y_1}{Y_1}$!(e$_{3}$) in! \\
            && !let A$_{0}$=shift1R (scanl (x,y.e$_{1}$) y$_{1}$ A) in! \\
            && !let A$_{1}$=shift1L (map2! \\
            && !   (a,b.(!$\grad_{\{x\}}$!e$_{1}$)[a/x,b/y]) A$_{0}$ A) in! \\
            && !let A$_{2}$=map2 (a,b.(!$\grad_{\{y\}}$!)e$_{1}$[a/x,b/y]) A$_{0}$ A in! \\
            && !let A$_{3}$=scanr * 1 A$_{1}$ in! \\
            && !<reduce (x,y.e$_{1}$) y$_{1}$ A, fun ($\boldsymbol{x}$,z) ->! \\
            && !Y$_{2}$($\boldsymbol{x}$, map2 (x,y. x*y*z) A$_{2}$ A$_{3}$)>! \\ \hline
        \end{tabular}
    \vspace{-0.4cm}
    \caption{Reverse-mode transformation from source to target language. We write 
    $\ad{\Gamma}{Y}$ instead of $\directD{\rho}{\Gamma}{Y}$, 
 and $\boldsymbol{x}$ instead of \texttt{x$_{1}$,$\ldots$,x$_n$}
    to alleviate notation.}
    \label{fig:direct_diff_macro}    
 
\end{figure*}

Admittedly, the transformation presented in this section may be hard to read, and it is not straightforward to show its correctness directly. 
Next, we decompose this transformation into three simpler steps via a novel intermediate representation.
\section{Unary Normal Form}
\label{sec:unf}

Following the intuition highlighted in Section~\ref{subsec:insights}, 
we present a new language, which we call Unary Normal Form (UNF). 
There is a Source UNF for the Source language and a Target UNF for the Target language.
Just as Target is an extension of Source, Target UNF is an extension of Source UNF.
They serve as intermediate representations in the reverse-mode compilation pipeline (see Figure~\ref{fig:outline}).

One takeaway from this paper is that efficient pure reverse-mode is complicated because of the several things it does.
It goes through the term and keeps a store for the gradients to be updated. 
It keeps track of the new bound variables found while going through the term.
The updates to the gradient are done via pre-composition and not post-composition.
Due to the linearity requirement in the usage of the continuation variable, 
it needs to process like a call-by-value evaluation and evaluate the arguments of operators in a certain order before 
dealing with the operator itself.

By introducing UNF, we are decoupling some of the problems.
UNF introduces the call-by-value evaluation of the arguments before evaluating the arguments and a good management of the environment.
Differentiation on UNF deals with the purely functional store and gradients update via pre-composition.
Finally, going back from UNF to Target uses the stored information in UNF to make sure the Jacobians are computed efficiently. 

\subsection{Source UNF} 
\label{sub:Source UNF}

Intuitively, a term consists of a composition of unary operators. 
To compile our source language to this intermediate representation, 
we need to remember some information about the context of the initial term. 
The grammar of Source UNF is given in Figure~\ref{fig:unf_source_grammar}. 

\begin{figure*}[t]
    \setlength{\tabcolsep}{0.3em}
    \centering
    \begin{tabular}{|l c l|l|}
    \hline
    \multicolumn{3}{|c|}{\textbf{Core Grammar}} & \multicolumn{1}{c|}{\textbf{Description}}\\\hline
    !T! & \mbox{::=} & ![A1,$\ldots$,An]! & \grammarcomment{Lists of types from source} \\
    \hline
    !e! & \mbox{::=} & !var!$_{T;i}$ & \grammarcomment{Variable}\\
    & $\mid$ & !op!$_{T;n}$ & \grammarcomment{Operations, for $0\leq n\leq 2$}\\
    & $\mid$ & !pair!$_{T;A\times B}$ & \grammarcomment{Pairing a pair of variables}\\
    & $\mid$ & !proj!$_{T_1;T_2;T_3}$  & \grammarcomment{Projection}\\
    & $\mid$ & !e!$\comp$!e! & \grammarcomment{Sequential composition}\\
    & $\mid$ & !map2!$_{T;x,y.e}$ & \grammarcomment{Map2}\\
    & $\mid$ & !reduce!$_{T;x,y.e;e}$ & \grammarcomment{Reduce}\\
    \hline
    \end{tabular}\\ \vspace{0.2cm}
    \begin{tabular}{|c|}
    \hline
    \begin{tabular}{c} 
    \\\hline
    !T! $\vdash$ !var!$_{T;i}$: !T,A$_i$!
    \end{tabular}~(!T=A$_{1}$,$\ldots$,A$_n$!)
    \hspace{0.5cm}
    \begin{tabular}{c}
        \\\hline
        !T,!$\reals^{\times(n)}$ $\vdash$ !op!$_{T;n}$ : !T,!$\reals^{\times(n+1)}$
    \end{tabular}~(!T=A$_{1}$,$\ldots$,A$_n$!)
\\
    \begin{tabular}{c}
    !T! $\vdash$ !e$_{1}$!: !T,A! $\quad$ !T,A! $\vdash$ !e$_{2}$!: !T,A,B! \\\hline
    !T! $\vdash$ !e$_{1}$!$\comp$!e$_{2}$!: !T,A,B!
    \end{tabular}~(!T=A$_{1}$,$\ldots$,A$_n$!)
\\
    \begin{tabular}{c}
        \\\hline
        !T,A,B! $\vdash$ !pair!$_{T;A\times B}$ : !T,A!$\times$!B!
    \end{tabular}~(!T=A$_{1}$,$\ldots$,A$_n$!)
\\
    \begin{tabular}{c}
        \\\hline
        !T$_{1}$,T$_{2}$,T$_{3}$! $\vdash$ !proj!$_{T_1;T_2;T_3}$ : !T$_{1}$,T$_{3}$!
    \end{tabular}
\\
    \begin{tabular}{c}
        !x$_{1}$:A$_{1}$,$\ldots$,x$_n$:A$_n$,x:$\reals$,y:$\reals$! $\vdash$ !e:$\reals$!\quad in Source Language
        \\\hline  
        !T,$\Array{\reals}{n},\Array{\reals}{n}$! $\vdash$ !map2$_{T; x,y.e}$: T,$\Array{\reals}{n}$,$\Array{\reals}{n}$,$\Array{\reals}{n}$!
    \end{tabular}~(!T=A$_{1}$,$\ldots$,A$_n$!)
\\
    \begin{tabular}{c}
        !x:$\reals$,y:$\reals$! $\vdash$ !e: $\reals$! \quad !x$_{1}$:A$_{1}$,$\ldots$,x$_n$:A$_n$! $\vdash$ !e$_{2}$!:$\reals$ \quad in Source Language
        \\\hline  
        !T,$\Array{\reals}{n}$! $\vdash$ !reduce$_{T; x,y.e_1; e_2}$: T,$\Array{\reals}{n}$,$\reals$!
    \end{tabular}~(!T=A$_{1}$,$\ldots$,A$_n$!) \\ \hline
\end{tabular}
    \vspace{-0.4cm}
    \caption{Grammar and type system of the source UNF.}
    \label{fig:unf_source_grammar}
    \vspace{-0.4cm}
    \end{figure*}

There are a few notable things in this syntax. 
Every primitive is indexed by a list of types from the source language, 
which corresponds to a context being carried. We will often elide the brackets ![!,!]!.
Every constant, unary, or binary operator from Source has a corresponding $n$-ary operator !op$_{T,n}$! in UNF.
Sequential composition is denoted by !;! and !e$_{1}$;e$_{2}$! means that !e$_{1}$! should be performed, and then !e$_{2}$!.
The array operators !map2! and !reduce! have extra indices that represent well-formed terms in Source.
The language also contains a projection !proj! which forgets some of the elements, 
and a pairing operator !pair! which takes the last pair of elements and returns the pairing of these elements.

A judgment in Source UNF is a triple !(T$_1$,e,T$_2$)! where !T$_1$! and !T$_2$! are lists of types of Source, and !e! is an expression of the Source UNF.
The typing rules are detailed in Figure~\ref{fig:unf_source_grammar}.
$\reals^{\times n}$ is a notation for $\reals,\ldots,\reals$ with $n$ factors $\reals$. 
When !T$_1$! and !T$_2$! are lists of types and !A! a type, the operator !,! denotes snoc in !T$_1$,A! and denotes append in !T$_1$,T$_2$!.

\begin{example}
    \label{exm:unf}
    The term in Source UNF !cos$_{\reals,\reals}$ $\comp$ pair$_{\reals,\reals; \reals\times\reals}$! 
    intuitively represents a term \\
    !x$_1$:$\reals$,x$_2$:$\reals$,x$_3$:$\reals$ $\vdash$ let (x$_1$,x$_2$,x$_3$,x$_4$)=(x$_1$,x$_2$,x$_3$,cos(x$_3$)) in (x$_1$,x$_2$,<x$_3$,x$_4$>)!.\\
    The operator $\comp$ can more generally be understood as a let binding of a tuple, and primitives only act on the last parts of the tuple, always also returning their whole context.
\end{example}

\subsection{Target UNF} 
\label{sub:Target UNF}

Target UNF is an extension of Source UNF. 
Its types are lists of types of Target. 
Higher-order types are needed for the continuation introduced by reverse-mode. 
In addition, it contains several new primitives which represent the transpose Jacobian of the primitives of Source UNF.
Given a type !T$\defeq$ [A$_1$,$\ldots$,A$_n$]!, we write \underline{T} for the type !A$_1$ $\times\ldots\times$ A$_n$! of Target. 
Target UNF has an internal composition $\icomp$. 
!e$_{1}$ $\icomp$ e$_{2}$! represents precomposition of a curried function !e$_{1}$! by a term !e$_{2}$!.
It is used by the reverse-mode transformation to pre-compose the continuation term by a transpose Jacobian.
The grammar and type system of target UNF is given in Figure~\ref{fig:unf_target_grammar}.

\begin{figure*}[t]
    \setlength{\tabcolsep}{0.3em}
    \centering
    \begin{tabular}{|l c l|l|}
    \hline
    \multicolumn{3}{|c|}{\textbf{Core Grammar}} & \multicolumn{1}{c|}{\textbf{Description}}\\\hline
    !T! & \mbox{::=} & ![A$_{1}$,$\ldots$,A$_n$]! & \grammarcomment{Lists of types from target} \\
    \hline
    !e! & \mbox{::=} & $\ldots$ & \grammarcomment{Same as source UNF}\\
    & $\mid$ & !J!$^T$!var!$_{T;i}$ & \grammarcomment{Jacobian for variable}\\
    & $\mid$ & !J!$^T$!op!$_{T;n}$ & \grammarcomment{Jacobian for operation, $0\leq n\leq 2$}\\
    & $\mid$ & !J!$^T$!pair$_{T;A\times B}$! & \grammarcomment{Jacobian for pairing}\\
    & $\mid$ & !J!$^T$!proj$_{T_1;T_2;T_3}$! & \grammarcomment{Jacobian for projection}\\
    & $\mid$ & !J!$^T$!map2!$_{T;x,y.e}$ & \grammarcomment{Jacobian for map2}\\
    & $\mid$ & !J!$^T$!reduce!$_{T;x,y.e;e}$ & \grammarcomment{Jacobian for reduce}\\
    & $\mid$ & !<e, e>! & \grammarcomment{Term pairing}\\
    & $\mid$ & !e! $\icomp$ !e! & \grammarcomment{Internal function composition}\\
    \hline
    \end{tabular} \\ \vspace{0.2cm}
    \begin{tabular}{|c|}
    \hline
    \begin{tabular}{c} 
        \\\hline
        !T,A$_i$! $\vdash$ !J!$^T$!var!$_{T;i}$: !T!
        \end{tabular}~(!T=A$_1$,$\ldots$,A$_n$!)
        \hspace{0.5cm}
        \begin{tabular}{c}
            \\\hline
            !T,!$\reals^{\times(n+1)} \vdash$ !J!$^T$!op!$_{T;n}$ : !T,!$\reals^{\times(n)}$
        \end{tabular}~(!T=A$_1$,$\ldots$,A$_n$!)
\\
        \begin{tabular}{c}
            \\\hline
            !T,A$\times$B! $\vdash$ !J!$^T$!pair!$_{T;A\times B}$ : !T,A,B!
        \end{tabular}~(!T=A$_1$,$\ldots$,A$_n$!)
\\
        \begin{tabular}{c}
            \\\hline
            !T$_1$,T$_3$! $\vdash$ !J!$^T$!proj!$_{T_1;T_2;T_3}$ : !T$_1$,T$_2$,T$_3$!
        \end{tabular}
\\  
        \begin{tabular}{c}
            !x$_1$:A$_1$,$\ldots$,x$_n$:A$_n$,x:!$\reals$!,y:!$\reals$ $\vdash$ !e!: $\reals$ \quad in Source Language
            \\\hline  
            !T,!$\Array{\reals}{n},\Array{\reals}{n},\Array{\reals}{n} \vdash$ !J!$^T$!map2!$_{T; x,y.e}$: !T,!$\Array{\reals}{n}$,$\Array{\reals}{n}$
        \end{tabular}~(!T=A$_1$,$\ldots$,A$_n$!)
\\  
        \begin{tabular}{c}
            !x:!$\reals$!,y:!$\reals$ $\vdash$ !e$_1$!: $\reals$ \quad !x$_1$:A$_1$,$\ldots$,x$_n$:A$_n$! $\vdash$ !e$_2$!:$\reals$ \quad in Source Language
            \\\hline  
            !T,!$\Array{\reals}{n},\reals \vdash$ !J!$^T$!reduce!$_{T; x,y.e_1; e_2}$: !T,!$\Array{\reals}{n}$
        \end{tabular}~(!T=A$_1$,$\ldots$,A$_n$!)
\\
        \begin{tabular}{c}
            !T! $\vdash$ !e$_1$: T$_1$!  \quad !T! $\vdash$ !e$_1$: T$_2$!
            \\ \hline
            !T! $\vdash$ !<e$_1$, e$_2$>: T$_1$,T$_2$!
        \end{tabular}
        \hspace{0.5cm}
        \begin{tabular}{c}
            !T1! $\vdash$ e$_1$: [\underline{T$_3$} $\to$ B]  \quad !T$_2$ $\vdash$ e$_2$: T$_3$! 
            \\ \hline
            !T1! $\vdash$ !e$_1$! $\icomp$ !e$_2$!:[\underline{T$_2$}! $\to$  B!]
        \end{tabular}\\\hline
    \end{tabular}
    \vspace{-0.4cm}
    \caption{Grammar and type system of Target UNF.}
    \label{fig:unf_target_grammar}
    \vspace{-0.4cm}
\end{figure*}

\subsection{Simple reverse mode transformation} 
\label{sub:Simple reverse mode transformation}

We are now able to present a simpler transformation for purely functional reverse-mode from Source UNF to Target UNF.
The differentiation transformation is given in Figure~\ref{fig:diff_macro}.
To a list of types, the transformation adds a higher-order type. 
It is to be understood as the continuation for the pure storage of the gradient.
On primitive constants of Source UNF, it returns a pair. 
The first part is the initial term (pre-composed by a projection because it is not using the continuation variable).
The second part returns the continuation term pre-composed by the transpose Jacobian of the primitive. 
There is a sharp contrast between the simplicity of $\Dsynrevsymbol_{\rho}$ compared to $\directD{\rho}{\Gamma}{Y}$ from Figure~\ref{fig:direct_diff_macro}.

\begin{example}
    Continuing with our previous example !cos$_{\reals,\reals}$ $\comp$ pair$_{\reals,\reals; \reals\times\reals}$!, we have

    \begin{tabular}{l}
        !$\reals$,$\reals$,$\reals$,$\reals\times\reals\times\reals$->$\rho \vdash \Dsynrevsymbol_{\rho}$(cos$_{\reals,\reals}\comp$ pair$_{\reals,\reals; \reals\times\reals}$): $\reals$,$\reals$,$\reals \times \reals$,$(\reals\times\reals\times(\reals\times\reals))$->$\rho$! \\
        = !<proj$_{\reals^3;\reals^3->\rho;[]} \comp$ cos$_{\reals,\reals}$, proj$_{[];\reals^3;\reals^3->\rho} \icomp$ (proj$_{\reals^3;\reals^3->\rho;[]}$ $\comp$ J$^T$cos$_{\reals,\reals}$)> $\comp$! \\
        !<proj$_{\reals^4;\reals^4->\rho;[]} \comp$ pair$_{\reals,\reals; \reals\times\reals}$, proj$_{[];\reals^4;\reals^4->\rho} \icomp$ (proj$_{\reals^4;\reals^4->\rho;[]}$ $\comp$ J$^T$pair$_{\reals,\reals; \reals\times\reals}$)>!
    \end{tabular}
\end{example}

\begin{lemma}[Well typedness of $\Dsynrevsymbol_{\rho}$]
    Let !A$_{1}$,$\ldots$,A$_{n}$ $\vdash$ e: B$_{1}$,$\ldots$,B$_{m}$! be a term in Source UNF. \\
    Then  !A$_{1}$,$\ldots$,A$_{n}$,A$_{1}\times\ldots\times$ A$_{n}$->$\rho$ $\vdash$ $\Dsynrevsymbol_{\rho}$(e): B$_{1}$,$\ldots$,B$_{m}$,B$_{1}\times\ldots\times$ B$_{m}$->$\rho$!.
\end{lemma}
\begin{proof}
    By induction on derivation of !A$_{1}$,$\ldots$,A$_{n}$ $\vdash$ e: B$_{1}$,$\ldots$,B$_{m}$!.
\end{proof}

\begin{figure*}[t]
\begin{tabular}{|c|}
\hline
    \begin{tabular}{r c l}
    $\Dsynrevsymbol_{\rho}$(!A$_1$,$\ldots$,A$_n$!) &=& !A$_1$,$\ldots$,A$_n$,(A$_1\times\ldots\times$A$_n$)->$\rho$!\\ \\
    $\Dsynrevsymbol_{\rho}$(!var!$_{T;i}$) &=& $F$(!var$_{T;i}$, J$^T$var$_{T;i}$!) \\
    $\Dsynrevsymbol_{\rho}$(!op!$_{\Gamma;n}$) &=& $F$(!op$_{T;n}$,J$^T$op$_{T;n}$!) \\ 
    $\Dsynrevsymbol_{\rho}$(!pair!$_{T;A\times B}$) &=& $F$(!pair$_{T;A\times B}$, J$^T$pair$_{T;A\times B}$!) \\
    $\Dsynrevsymbol_{\rho}$(!proj!$_{T_1;T_2;T_3}$) &=& $F$(!proj$_{T_1;T_2;T_3}$,J$^T$proj$_{T1;T2;T3}$!) \\
    $\Dsynrevsymbol_{\rho}$(!e$_1\comp$e$_2$!) &=& $\Dsynrevsymbol_{\rho}$(!e$_1$!)$\comp$ $\Dsynrevsymbol_{\rho}$(!e$_2$!)\\ 
    $\Dsynrevsymbol_{\rho}$(!map2!$_{T;x,y.e}$) &=& $F$(!map2$_{T;x,y.e}$, J$^T$map2$_{T;x,y.e}$!) \\
    $\Dsynrevsymbol_{\rho}$(!reduce!$_{T;x,y.e_1;e_2}$) &=& $F$(!reduce$_{T;x,y.e_1;e_2}$, J$^T$reduce$_{T;x,y.e_1;e_2}$!) \\
    \end{tabular}\\
    where !$F$(A,B)$\defeq$ <proj$_{T;\underline{T}->\rho;[]} \comp$ A, proj$_{[];T;\underline{T}->\rho} \icomp$ (proj$_{T;\underline{T}->\rho;[]}$ $\comp$ B)>!\\\hline
    \end{tabular}
    \vspace{-0.4cm}
    \caption{Reverse-mode differentiation from Source UNF to Target UNF}
    \label{fig:diff_macro}    
    \vspace{-0.4cm}
\end{figure*}

\subsection{Transformations to and from UNF} 
\label{sub:transformations to and from UNF}

We first give a translation from our Source to Source UNF in Figure~\ref{fig:source_to_unf}, which we call $\UNFSymbol$.
As discussed above, $\UNFSymbol$ sends a constant or an operator to a primitive in Source UNF. 
This term in Source UNF carries the context of the original term in Source.
$\UNFSymbol$ sequentializes a term, mimicking a left-to-right call-by-value evaluation.
Because of this, when we are trying to return a pair, this needs to be witnessed, and this is the role of !pair!.
When we need to pass the result of type !A! of !e$_{1}$! through !e$_{2}$! which does not take it as an input, we can transform !e$_{2}$! to accept and pass it.
It's the equivalent in UNF of a weakening.
This weakening is required when we transform a non-unary operator, such as $+$ or !pair!, as we transform the arguments in order, and we need to keep all of their results.
We write this new term $(\widetilde{e_{2}})_{A}$. It is defined by induction on !e$_{2}$! as follows. 
We assume given a type !A! that needs to be passed, and drop the index !A!.

\begin{tabular}{r c l}
    $\widetilde{var_{T;i}}$  &=& $var_{T,A;i}$ \\
    $\widetilde{op_{T;n}}$  &=& $op_{T,A;n}$ \\
    $\widetilde{pair_{T;B\times C}}$  &=& $pair_{T,A;B\times C}$ \\
    $\widetilde{proj_{T_1;T_2;T_3}}$  &=& $proj_{T_1;T_2;T_3,A}$ 
\end{tabular}
\begin{tabular}{r c l}
    $\widetilde{e_1 \comp e_2}$  &=& $\widetilde{e_1}\comp\widetilde{e_2}$ \\
    $\widetilde{map2_{T;x,y.e}}$  &=& $map2_{T,A;x,y.e}$ \\
    $\widetilde{reduce_{T;x,y.e_1;e_2}}$ &=& $reduce_{T,A;x,y.e_1;e_2}$
\end{tabular}

We want to preserve the invariant that $\UNFSymbol$(!e!) represents the term !e! which is also returning its context. 
Because !x! is not free in $\Gamma \vdash$ !let x=e$_1$in e$_2$!, we need to hide !x! after $\UNFSymbol$(!e$_2$!).
This explains the projection !proj! in $\UNFSymbol$ of a !let!. 
More generally, in !e!$_{1} \comp$ !e!$_{2}$ the return value of !e!$_{1}$ acts as a new variable for !e!$_{2}$, 
is not in the context of !e!$_{1}$, but will be returned by !e!$_{2}$.
This means we need to hide all these intermediate values to preserve the invariant,
and this explains the projection !proj! in the $\UNFSymbol$ of other terms. 

\begin{lemma}[Well typedness of $\UNFSymbol$]
    Let !$\Gamma$=x$_1$:A$_{1}$,$\ldots$,x$_n$:A$_{n}$ $\vdash$ e:B! be a term in Source. Then
    \begin{center}
        !A$_{1}$,$\ldots$,A$_{n}$ $\vdash$ $\UNFSymbol$(e): A$_{1}$,$\ldots$,A$_{n}$,B!.
    \end{center}
\end{lemma}
\begin{proof}
    By induction on derivation of $\Gamma\vdash$!e:B!.
\end{proof}

\begin{figure*}[t]
    \begin{tabular}{|r c l|}
        \hline
    $\UNFSymbol$($\Gamma\vdash $ !c!) &=& !c!$_{\Gamma,0}$ constant seen as a 0-ary operator\\
    $\UNFSymbol$($\Gamma\vdash $ !x!) &=& !var!$_{\Gamma,i}$ where !x! is the $i$-th variable in $\Gamma$ \\
    $\UNFSymbol$($\Gamma\vdash $ !let x:A = e$_1$ in e$_2$:B!) &=& $\UNFSymbol$(!e$_1$!) $\comp$ $\UNFSymbol$(!e$_2$!) $\comp$ !proj!$_{\Gamma;A;B}$  \\ 
    $\UNFSymbol$($\Gamma\vdash $ !< e$_1$, e$_2$ >:AxB!) &=& $\UNFSymbol$(!e$_1$!) $\comp$ $\widetilde{\UNFSymbol(e_2)}$ $\comp$ !pair!$_{\Gamma;A\times B}$ \\ 
    $\UNFSymbol$($\Gamma\vdash \pi_i$(!e!)) &=& $\UNFSymbol$(!e!)$\comp $ $\pi_i$ $\comp$ !proj!$_{\Gamma;A_1\times A_2;A_i}$ \\
    && $\pi_i$ seen as a unary operator\\
    $\UNFSymbol$($\Gamma\vdash $ !e$_1$ op2 e$_2$!) &=& $\UNFSymbol$(!e$_1$!) $\comp$ $\widetilde{\UNFSymbol(e_2)}$ $\comp$ !op!$_{\Gamma;2}$ $\comp$ !proj!$_{\Gamma;\reals,\reals;\reals}$ \\
    $\UNFSymbol$($\Gamma\vdash $ !op1 e!) &=& $\UNFSymbol$(!e!) $\comp$ !op!$_{\Gamma;1}$ $\comp$ !proj!$_{\Gamma;\reals;\reals}$ \\
    $\UNFSymbol$($\Gamma\vdash $ !map2 (x,y.e$_1$) e$_2$ e$_3$!) &=& $\UNFSymbol$(!e$_2$!) $\comp$ $\widetilde{\UNFSymbol(e_3)}$ $\comp$ !map2!$_{\Gamma; x,y.e_1}$ $\comp$ \\
    && !proj!$_{\Gamma;\Array{\reals}{n},\Array{\reals}{n};\Array{\reals}{n}}$ \\ 
    $\UNFSymbol$($\Gamma\vdash $ !reduce (x,y.e$_1$) e$_2$ e$_3$!) &=& $\UNFSymbol$(!e$_3$!)$\comp$ !reduce!$_{\Gamma; x,y.e_1; e_2}$ $\comp$ !proj!$_{\Gamma;\Array{\reals}{n};\Array{\reals}{n}}$ \\
    \hline 
    \end{tabular}
    \vspace{-0.4cm}
    \caption{Transformation from Source to Source UNF}
    \label{fig:source_to_unf}
    \vspace{-0.4cm}
    \end{figure*}

Next, the transformation from Target UNF to Target is presented in Figure~\ref{fig:unf_to_target}. 
We call this transformation $\UNFSymbol^{-1}$, but it is not a strict inverse of $\UNFSymbol$.
Doing $\UNFSymbol$ followed by $\UNFSymbol^{-1}$ performs some version of the ANF transformation \cite{sabry1993reasoning}.

Target UNF does not have variables. 
A context type !T=[A$_1$,$\ldots$,A$_n$]! is transformed to a context \\
!$\Gamma$=x$_1$:A$_1$,$\ldots$,x$_n$:A$_n$!.
We use this convention when definiting $\UNFSymbol^{-1}$. 
For operators like !op$_{\Gamma;m}$!, we extend the notation above by saying that the variables from the context are \\
!x$_1$:A$_1$,$\ldots$,x$_n$:A$_n$,x$_{n+1}$:$\reals$,$\ldots$,x$_{n+m}$:$\reals$!.



$\UNFSymbol^{-1}$ does not have as a simple typing property as $\UNFSymbol$, 
because it treats primitives with and without a $J$ differently and should be performed after $\Dsynrevsymbol_{\rho}$.

\begin{lemma}[Well typedness of $\UNFSymbol^{-1}$]
    Let !A$_{1}$,$\ldots$,A$_{n}$ $\vdash$ e: B$_{1}$,$\ldots$,B$_{m}$! be a term in Source UNF. Then
    \begin{center}
        !x$_1$:A$_{1}$,$\ldots$,x$_n$:A$_{n}$,x$_{n+1}$:A$_{1}\times\ldots\times$ A$_{n}$->$\rho$ $\vdash$ $\UNFSymbol^{-1}$($\Dsynrevsymbol_{\rho}$(e)): B$_m$ $\times$ (B$_{1}\times\ldots\times$ B$_{m}$->$\rho$) !.
    \end{center}
\end{lemma}
\begin{proof}
    By induction on derivation of !A$_{1}$,$\ldots$,A$_{n}$ $\vdash$ e: B$_{1}$,$\ldots$,B$_{m}$!.
\end{proof}

\begin{figure*}[t]
    \begin{tabular}{|r c l|}
    \hline
    $\invUNFSymbol$(!var!$_{\Gamma;i}$) &=& !x$_i$: Ti! \\
    $\invUNFSymbol$(!op!$_{\Gamma;m}$) &=& !op$_n$(x$_n$,$\ldots$,x$_{n+m}$)! \\
    $\invUNFSymbol$(!e$_1$ $\comp$ e$_2$!) &=& !let x$_{n+1}$=$\invUNFSymbol$(e$_1$) in $\invUNFSymbol$(e$_2$)! \\
    $\invUNFSymbol$(!map2!$_{\Gamma;x,y.e}$) &=& !map2 (x,y.e) x$_{n+1}$ x$_{n+2}$! \\
    $\invUNFSymbol$(!reduce!$_{\Gamma;x,y.e_1;e_2}$) &=& !reduce (x,y.e$_1$) $e_2$ x$_{n+1}$! \\
    $\invUNFSymbol$(!< e$_1$, e$_2$>!) &=& !<$\invUNFSymbol$(e$_1$), $\invUNFSymbol$(e$_2$)>! \\
    $\invUNFSymbol$(!proj!$_{T_1;T_2;T_3}$) &=& !(x$_{k}$,$\ldots$,y$_{p}$)! where the !x!$_{i}$ are the variables of $T_3$ \\
    $\invUNFSymbol$(!pair!$_{T;A\times B}$) &=& !<x$_{n+1}$, x$_{n+2}$>! \\
    $\invUNFSymbol$(!e1!$\icomp$!e2!) &=& !fun (y$_1$,$\ldots$,y$_m$) -> $\invUNFSymbol$(e$_1$)($\invUNFSymbol$(e$_2$)[$\forall i$,y$_i$/x$_i$])! \\
    $\invUNFSymbol$(!J!$^T$!var!$_{\Gamma;i}$) &=& !(x$_1$,$\ldots$,x$_{i-i}$,x$_i$+x$_{n+1}$,x$_{i+1}$,$\ldots$,x$_{n}$)! \\
    $\invUNFSymbol$(!J!$^T$!op!$_{\Gamma;m}$) &=& !(x$_1$,$\ldots$,x$_{n}$,x$_{n+1}$+$\partial_1$op$_n$*x$_{n+m+1}$,$\ldots$,x$_{n+m}$+$\partial_m$op$_n$*x$_{n+m+1}$)! \\
    $\invUNFSymbol$(!J!$^T$!proj!$_{T_1;T_2;T_3}$) &=& !(x$_1$,$\ldots$,x$_k$,0,$\ldots$,0,x$_{k+p}$,$\ldots$,x$_n$)! \\
    $\invUNFSymbol$(!J!$^T$!pair!$_{T;A\times B}$) &=& !(x$_1$,$\ldots$,x$_{n-1}$,$\pi_1$x$_n$,$\pi_2$x$_n$)! \\
    $\invUNFSymbol$(!J!$^T$!map2!$_{\Gamma;x,y.e}$) &=& 
    !let G = (map2 * x$_{n+3}$ !\\
    && \quad\quad!(map2 (a,b.(!$\grad_\Gamma$!e$_{1}$)[a/x.b/y])  x$_{n+1}$  x$_{n+2}$)) in! \\
    && !((x$_{1}$,$\ldots$,x$_n$)$\widehat{+}$(reduce $\widehat{+}$ $\widehat{0}$ G),!\\
    && ! map2 (a,b.(!$\grad_{\{x\}}$!e$_{1}$)[a/x]*b) x$_{n+1}$ x$_{n+3}$,!\\
    && ! map2 (a,b.(!$\grad_{\{y\}}$!e$_{1}$)[a/x]*b) x$_{n+2}$ x$_{n+3}$ )!\\
    $\invUNFSymbol$(!J!$^T$!reduce!$_{\Gamma;x,y.e1;e2}$) &=& !let A$_{0}$ = shift1R (scanl (x,y.e$_{1}$) y$_{1}$ x$_{n+1}$) in! \\
    && !let A$_{1}$ = shift1L (map2 !\\
    && !     (a,b.(!$\grad_{\{x\}}$!e$_{1}$)[a/x,b/y]) A$_{0}$ x$_{n+1}$) in! \\
    && !let A$_{2}$ = map2 (a,b.(!$\grad_{\{y\}}$!)e$_{1}$[a/x,b/y]) A$_{0}$ x$_{n+1}$ in! \\
    && !let A$_{3}$ = scanr * 1 A$_{1}$ in! \\
    && !(x$_{1}$,$\ldots$,x$_n$, map2 (x,y. x*y*x$_{n+2}$) A$_{2}$ A$_{3}$)! \\
    \hline
    \end{tabular}
    \vspace{-0.4cm}
    \caption{Transformation from Target UNF to Target}
    \label{fig:unf_to_target}
    \vspace{-0.4cm}
    \end{figure*}

\section{Complexity analysis}
\label{sec:complexity}

In this section, we introduce a simple cost model for our language and show 
that, after partial evaluation, our reverse-mode transformation satisfies a version of the cheap gradient principle (Theorem \ref{thm:complexity}).

\subsection{Partial evaluation and optimization} 
\label{sub:Partial evaluation and optimization}

As can be seen from the examples, $\directD{\rho}{\Gamma}{Y}$ introduces a lot of functions of several arguments.
Following the insight from Section~\ref{subsec:insights}, we do not want to keep all these costly lambda abstractions.
The transformation is designed in such a way that all the lambda abstractions are given arguments,  
and we can use partial evaluation to beta-reduce all these lambda abstractions. 
By inspecting each rule in Figure~\ref{fig:direct_diff_macro}, this allows us to prove by induction on the judgment of !e!:

\begin{lemma}
    Let $\Gamma \vdash$!e: A! be a term in Source. 
    Every variable !Y! in $\directD{\rho}{\Gamma}{Y}$(!e!) which is not of ground type has exactly one occurrence in the term.
\end{lemma}

From the previous lemma, we see that there is a linear usage of each continuation variable $Y$. 
This is one key property ensured by $\UNFSymbol$.
We also note that, apart from the !map2! case, this continuation variable is always applied to almost the identity. 
More precisely, we have applications of the form !fun (x$_1$,$\ldots$,x$_n$) -> Y(e$_1$,$\ldots$,e$_n$)! 
where the !e$_i$! are !x$_i$! except for at most $k$ (independent of $n$) terms.
In fact, we have $k=2$, except for the !map2! and !reduce!.

We can first use inlining, the first optimization rule given in Figure~\ref{fig:optim}.
Using the invariant above, most of the !e!$_{i}$ below are variables. 
Without loss of generality, assume the only terms !e!$_{i}$ which are not variables are !e!$_{n-1}$ and !e!$_{n}$.
We can then use forward substitution, the third optimization rule in Figure~\ref{fig:optim}, on the other !e!$_{i}$.
In brief:

\begin{tabular}{l}
!fun (x$_{1}$,$\ldots$,x$_n$) -> (fun (y$_{1}$,$\ldots$,y$_n$) -> (f$_{1}$,$\ldots$,f$_n$))(e$_{1}$,$\ldots$,e$_n$)! 
$\transto$ \\
!fun (x$_{1}$,$\ldots$,x$_n$) -> let y$_{1}$,$\ldots$,y$_n$ = e$_1$,$\ldots$,e$_n$ in (f$_{1}$,$\ldots$,f$_n$)! $\transto$ \\
!fun (x$_{1}$,$\ldots$,x$_n$) -> let y$_{n-1}$,y$_n$ = e$_{n-1},$e$_1$ in!\\
\hspace{3cm}!(f$_{1}$[x$_{1}$/y$_{1}$],$\ldots$,f$_{n-2}[$x$_{n-2}$/y$_{n-2}]$,f$_{n-1}$,f$_{n}$)!
\end{tabular}

This rewriting does not change the evaluation cost of the !f!$_{i}$ for $1\leq i \leq n-2$.
The new evaluation cost is reduced to the sum of the cost of evaluating 
the !f!$_{i}$ in addition to the cost of evaluating !e!$_{n-1}$ and !e!$_{n}$, gaining O($n$) movement of variables.

\begin{example}
    After forward-substitution and inlining the inner !Y!$_{i}$, the gradient of the terms from the introduction reduces to

    \begin{tabular}{c l}
        & $\grad_\Gamma$(!let w$_{1}$ = x$_{1}$ * x$_{2}$ in let w$_{2}$ = w$_{1}$ * x$_{1}$ in w$_{2}$!) \\
        =& !let w$_1$,Y$_1$= <x$_1$ * x$_{2}$, fun (y$_{1}$,y$_{2}$,y$_{3}$,z) -> Y(y$_{1}$+x$_2$*z,y$_{2}$+x$_1$*z,y$_{3}$) > in! \\
        & !let w$_{2}$,Y$_{2}$= <w$_{1}$ * x$_{1}$, fun (y$_{1}$,y$_{2}$,y$_{3}$,y$_{4}$,z) -> Y$_{1}$(y$_{1}$+w$_1$*z,y$_{2}$,y$_{3}$,y$_{4}$+x$_{1}$*z) > in! \\
        & !let y,Y$_{3}$= <w$_{2}$, fun (y$_{1}$,y$_{2}$,y$_{3}$,y$_{4}$,z) -> Y$_{2}$(y$_{1}$,y$_{2}$,y$_{3}$,y$_{4}$+z)> in! \\
        & !<y, fun (y$_{1}$,y$_{2}$,y$_{3}$,z) -> Y$_{3}$(y$_{1}$,y$_{2}$,y$_{3}$,0,z) >!
    \end{tabular}

After another simplification step, we obtain 

\begin{tabular}{c l}
    & $\grad_\Gamma$(!let w$_{1}$ = x$_{1}$ * x$_{2}$ in let w$_{2}$ = w$_{1}$ * x$_{1}$ in w$_{2}$!) \\
    =& !let w$_1$= x$_1$ * x$_{2}$ in let w$_{2}$= w$_{1}$ * x$_{1}$ in! \\
        & !<w$_{2}$, fun (y$_{1}$,y$_{2}$,y$_{3}$,z) -> let y'$_{1}$=y$_{1}$+w$_1$*z in !\\
        & \quad\quad\quad!let z'=y$_4$+x$_1$*z in Y(y'+x$_2$*z',y$_{2}$+x$_1$*z',y$_{3}$) >!
\end{tabular}

Similarly, for the gradient of !prod(A)! we obtain

\begin{tabular}{c l}
        & $\grad_\Gamma$(!prod(A)!) \\
        & !let A$_{0}$= shift1R (scanl * 1 A) in!\\
        & !let A$_{1}$= shift1L (map2 (a,b.b) A$_{0}$ A) in!\\
        & !let A$_{2}$= map2 (a,b.a) A$_{0}$ A in!\\
        & !let A$_{3}$= scanr * 1 A$_{1}$ in!\\
        & !<prod(A), fun (X,z) -> Y(X+map2 (a,b. a*z*b) A$_{2}$ A$_{3}$)>! 
    \end{tabular}
\end{example}

We call this optimization step partial evaluation in the rest of the complexity section. 
We have the following result.

\begin{lemma}
    \label{lem:noHO}
    Let $\Gamma \vdash$ !e:A! be a term in Source. 
    After the partial evaluation step, $\directD{\rho}{\Gamma}{Y}$(!e!) does not have lambda abstractions.
    Its only variable (bound and free) which does not have a ground type is !Y!.
\end{lemma}

\begin{proof}
    By induction on the judgment of !e!. 
    We inspect each rule in Figure~\ref{fig:direct_diff_macro} 
    and note that the partial evaluation step precisely allows us to conclude the inductive step.
\end{proof}

\subsection{Cost model}
\label{sub:costModel}

We follow a simple model similar to the one in \cite{griewank2008evaluating}.
We assume the cost is divided into $4$ elementary measures, being the number of MOVES, ADDS, MULTS, and NLOPS.
MOVES assumes a flat memory and represents moving fixed-size information (e.g. 64 bits). 
ADDS represents the number of additions, 
MULTS the number of multiplications, 
and NLOPS the number of elementary non-linear operations like !cos! or !exp!.

The gives a complexity function $\cost$ valued in $\RR^4$. 
For primitive operations, we have for instance

\begin{tabular}{ll}
    $\cost$(!*!)=$(3,0,1,0)$ & $\cost($!+!$)=(3,1,0,0)$\\
    $\cost$(!c!)=$(1,0,0,0)$ & $\cost($!sin!$)=(2,0,0,1)$
\end{tabular}

More generally $\cost($!op1!$)=(2,0,0,1)$ for the other unary operations.
For simplicity, we will not address any parallelism in our cost model. 
So the cost of !map (x.e)! on an array of size $n$ will be $n*\cost$(!e!).
Following Lemma \ref{lem:noHO}, we do not need our cost model to deal with higher-order variables and lambda abstractions.
We can thus restrict our attention to the subset of the Target language which does not contain lambdas, applications or variables which are not of ground type.
The cost function extends compositionally to the restricted Target language. 
It is given in Figure~\ref{fig:costmodel}.

\begin{figure*}[t]
    \begin{tabular}{|l c l|}
        \hline
      !$\cost$(c)! &=& !(1,0,0,0)! \\ 
      !$\cost$(x)! &=& !(1,0,0,0)! \\ 
      !$\cost$(op1 e)! &=& !$\cost$(op1)+$\cost$(e)! \\ 
      !$\cost$(e$_1$ op2 e$_2$)! &=& !$\cost$(op2)+$\cost$(e$_1$)+$\cost$(e$_2$)! \\ 
      !$\cost$(<e$_1$,e$_2$>)! &=& !$\cost$(e$_1$)+$\cost$(e$_2$)! \\ 
      !$\cost$($\pi_i$(e))! &=& !(1,0,0,0)+$\cost$(e)! \\ 
      !$\cost$(let x=e$_1$ in e$_2$)! &=& !(1,0,0,0)+$\cost$(e$_1$)+$\cost$(e$_2$)! \\ 
      !$\cost$(map2 (x,y.e$_1$) e$_2$ e$_3$)! &=& !n*($\cost$(e$_1$)+(2,0,0,0))+$\cost$(e$_2$)+$\cost$(e$_3$)! \\ 
      !$\cost$(reduce (x,y.e$_1$) e$_2$ e$_3$)! &=& !n*($\cost$(e$_1$)+(2,0,0,0))+$\cost$(e$_2$)+$\cost$(e$_3$)! \\ 
      !$\cost$(scanl (x,y.e$_1$) e$_2$ e$_3$)! &=& !n*($\cost$(e$_1$)+(3,0,0,0))+$\cost$(e$_2$)+$\cost$(e$_3$)! \\ 
      !$\cost$(scanr (x,y.e$_1$) e$_2$ e$_3$)! &=& !n*($\cost$(e$_1$)+(3,0,0,0))+$\cost$(e$_2$)+$\cost$(e$_3$)! \\ 
      !$\cost$(shift1L e)! &=& !$\cost$(e)+(n,0,0,0)! \\ 
      !$\cost$(shift1R e)! &=& !$\cost$(e)+(n,0,0,0)! \\ 
      \hline
    \end{tabular}
    \vspace{-0.4cm}
    \caption{Cost model for the restricted target language}
    \label{fig:costmodel}
    \vspace{-0.4cm}
\end{figure*}

\subsection{Cheap gradient principle}

We define the Nesting of Array Operations $\NestA$ of a term !e! of Source by induction on !e! as follows.

\begin{center}
\begin{tabular}{r c l}
    $\NestA$(!c!), $\NestA$(!x!) &=& 0 \\
    $\NestA$($\pi_i$(!e!)), $\NestA$(!op1 e!) &=& $\NestA$(!e!) \\
    $\NestA$(!let x:A = e$_1$ in e$_2$:B!) &=& max($\NestA$(!e$_1$!), $\NestA$(!e$_2$!))  \\ 
    $\NestA$(!< e$_1$, e$_2$ >:AxB!) &=& max($\NestA$(!e$_1$!), $\NestA$(!e$_2$!)) \\ 
    $\NestA$(!e$_1$ op2 e$_2$!) &=& max($\NestA$(!e$_1$!), $\NestA$(!e$_2$!))\\
    $\NestA$(!map2 (x,y.e$_1$) e$_2$ e$_3$!) &=& max(1+$\NestA$(!e$_1$!), $\NestA$(!e$_2$!), $\NestA$(!e$_3$!)) \\
    $\NestA$(!reduce (x,y.e$_1$) e$_2$ e$_3$!) &=& max(1+$\NestA$(!e$_1$!), $\NestA$(!e$_2$!), $\NestA$(!e$_3$!))
\end{tabular}
\end{center}

We can now phrase our main complexity theorem.

\begin{theorem}
    \label{thm:complexity}
    Given a term !$\Gamma \vdash$ e: $\reals$! such that $\NestA$(!e!)$\leq$ $p$.
    Denote by $G$ the term $\grad_\Gamma$!e! after the partial evaluation step from Section~\ref{sub:Partial evaluation and optimization}.
    Then !$\cost$(G) $\leq$ $4*3^{p}*\cost$(e)!.
\end{theorem}

The cheap gradient principle (see e.g. \cite{griewank2008evaluating}) for reverse-mode
asserts that evaluating the gradient of a function $e:\RR^n\to\RR$ 
should be the same order of cost as evaluating $e$. 
More precisely, there should be a constant $K$ such that for each program !$\Gamma \vdash$ e: $\reals$! in the context $\Gamma=\{x_1:\reals,\ldots,x_n:\reals\}$,
!$\cost$($\grad$e) $\leq$ $K$*$\cost$(e)!.

\begin{proof}[Proof Sketch]
    As $\directD{\rho}{\Gamma}{Y}$ is defined by induction on programs, it suffices to show locally 
that $\directD{\rho}{\Gamma}{Y}$ verifies the cheap gradient principle. 

There is a first restriction that prevents our transformation from satisfying the cheap gradient principle.
If we try to show that the cheap gradient principle holds by induction on terms, it fails for !map2!.
When differentiating !map2!, there are three series of calls to (sub)gradients of !e$_1$!. 
The induction hypothesis is then too weak to conclude. 
The problem comes from the fact that !e$_1$! could itself use !map2!. 
So the constant $K$ can be independent of $n$ and of 
the size of the term if we allow it to be dependent on the level of nesting of !map2!.
A similar phenomenon happens with !reduce!.

Another problem is that the continuation part of $\directD{\rho}{\Gamma}{Y}$ adds $O(n)$ MOVES at each step.
This overhead is precisely what that our inlining and forward substitution removes, 
as was exemplified in Section~\ref{sub:Partial evaluation and optimization}. 

After that, the proof is by routine induction on !e!. 
First, one computes the cost of $\directD{\rho}{\Gamma}{Y}$(!e!).
It has $n$ too many MOVES which are removed by inlining and forward substitution using the invariant that at each step, 
at most 2 !x!$_i$ in the continuation are not variables.
\end{proof}

\begin{figure*}[t]
    \begin{tabular}{|l c l|}
        \hline
        \textit{Inlining and forward substitution}  & &\\ \hline
        !(fun (x$_1$,$\ldots$,x$_n$) -> e)(e$_1$,$\ldots$,e$_n$)! & \multirow{2}{*}{\transto} & !let x$_1$=e$_1$ in $\ldots$! \\
        && !let x$_n$=e$_n$ in e! \\ \hline
        !let x$_1$=e$_1$ in e$_2$! \quad(!x$_1$!$\not\in$!FV(e$_2$))! & \transto & !e$_2$!  \\ \hline
        !let x$_1$=x$_2$ in e! & \transto & !e[x$_2$/x$_1$]! \\ \hline
        !let x$_1$= c in e!  \quad(!c= 0,1!) & \transto & !e[c/x$_1$]! \\
        \hline \hline
        \textit{Algebraic simplifications}  & & \\ \hline
        !0*e! & \transto & 0 \\ \hline
        !0+e, 1*e! & \transto & !e! \\
        \hline \hline
        \textit{Array algebraic simplifications}  & & \\ \hline
        !map2 * A OnesLike(B)!  & \multirow{3}{*}{\transto} & \\
        !map2 + A ZerosLike(B)! && !A!\\
        !map (x.x) A! && \\ \hline
        !map2 * A ZerosLike(B)! & \transto & !ZerosLike(B)! \\ \hline
        !reduce * 1 OnesLike(A)! & \transto & !1! \\ \hline
        !reduce + 0 ZerosLike(A)! & \transto & !0! \\ \hline
        !shift1L OnesLike(n+1)! & \multirow{2}{*}{\transto} & !OnesLike(n)! \\ 
        !shift1R OnesLike(n+1)! && \\ \hline
       
        \hline \hline
        \textit{Classic array simplification}  & & \\ \hline
        !map (x.e$_1$) (map2 (y$_1$,y$_2$.e$_2$) A B)! & \transto & !map2 (y$_1$,y$_2$.let x=e$_2$ in e$_1$) A B! \\ \hline
        \textit{Tuple partial evaluation}  & & \\ \hline
        $\pi_i$!<e$_1$,$\ldots$,e$_n$>! & \transto & !e$_i$! \\
        \hline \hline
        \textit{Let normalisation}  & & \\ \hline
        !let x=(let y=e$_1$ in e$_2$) in e$_3$! & \transto & !let y=e$_1$ in let x=e$_2$ in e$_3$! \\ \hline
        !f(let x=e$_1$ in e$_2$)! & \transto & !let x=e$_1$ in f(e$_2$)! \\
        \hline \hline
        \textit{Conditionals} & & \\ \hline
        !if e$_1$ then e$_2$ else e$_2$! & \transto & e$_2$ \\ \hline
        !if true then e$_2$ else e$_3$! & \transto & e$_2$ \\ \hline
        !if false then e$_2$ else e$_3$! & \transto & e$_3$ \\ \hline
        !f(if e$_1$ then e$_2$ else e$_3$)! & \transto & !if e$_1$ then f(e$_2$) else f(e$_3$)! \\ \hline
        \end{tabular}
    \vspace{-0.4cm}
    \caption{Optimizations for target language.}
    \label{fig:optim} 
    \vspace{-0.4cm}
\end{figure*}

\subsection{Optimizations} 
\label{sub:Optimizations}

In Figure~\ref{fig:optim} we present a list of optimizations. 
These optimizations are all obviously valid according to the semantic model. 
As our language is purely functional, we can use these extra optimizations aggressively. 
A lot of simplifications come from the ring structure of the reals, lifted to tuples and arrays.
After these optimizations, the output program is essentially a sequence of let-bindings 
and resembles SSA-code \cite{cytron1989efficient}.

Even though our reverse-mode transformation has the right complexity after the partial evaluation step, 
the constant factor is of huge importance in practice. 
The purity of our transformation allows us to make the most out of generic optimizations.  
In addition, hand-crafted efficient derivatives and more optimizations can easily be added to our language.

\begin{example}
    As shown in the supplementary material,  
    the optimizations from Figure~\ref{fig:optim} 
    are sufficient to show that the gradients of the terms of the introduction reduce to the following.
    
    \begin{tabular}{{r c l}}
        $\nabla_A$!prod(A)! &=& !map2 * (scanr * 1 (shift1L A)) (shift1R (scanl * 1 A))!\\
        $\nabla_A$!sum(A)! &=& !map (x -> 1) A!\\
        $\nabla_A$!dot(A,B)! &=& !B! 
    \end{tabular}
\end{example}
\section{Correctness}
\label{sec:correctness}

We now explain that the correctness of the three steps of the $\directD{\rho}{\Gamma}{Y}$ transformation from Section~\ref{sub:Macro for pure reverse mode transformation} can be understood in terms of translating between different standard categorical languages, 
and verified by working in a well-known category for differentiability: diffeological spaces. 

In general, categories with products (and more generally monoidal categories) can be equivalently presented in the following three ways. 
Although the three styles of presentation are categorically equivalent, the choice of presentation affects the internal language syntax.
\begin{itemize}
    \item categories, where each arrow has one source and one target. These are most familiar, but do not directly match with a programming syntax (except the categorical abstract machine). 
    \item multicategories, where each arrow has a list of source objects !A$_1,\ldots,$A$_n$! and one target !B!, thought of as a map !A$_1$ $\times\ldots\times$A$_n\to$B!. 
    These are very close to the syntax usually used in typed programming languages \cite{lambek1968deductive,staton2013universal}, 
    thinking of an arrow as a typed term !x$_1$:A$_1,\ldots,$x$_n$:A$_n$ $\vdash$ B!.
    \item concategories, aka coloured props \cite{bonchi2015full,fong2019backprop}, where each arrow has a list of source objects !A$_1\ldots$A$_n$! and a list of target objects !B$_1\ldots$B$_m$!, thought of as a map !A$_1\times\ldots\times$A$_m \to$ B$_1\times\ldots\times$B$_m$!. These match the syntax of the UNF language.
\end{itemize}

Informally, notice that it makes sense to talk about the opposite of a category or a concategory, but not the opposite of a multicategory. 
For roughly this reason, concategories are a more natural place to consider reverse derivatives than multicategories, 
even though they are syntactically less familiar. 
The $\directD{\rho}{\Gamma}{Y}$ translation passes from the multicategory for the source syntax to the concategory for UNF, where reverse-mode differentiation is easier, and then back to the multicategory for the target syntax. 
In this way, the translations between Source/Target and Source UNF/Target UNF are merely changing perspective from multicategories to concategories.
We work with a specific multicategory and concategory built from diffeological spaces and smooth maps, so that we can verify the construction of reverse derivatives. 
Since diffeological spaces support products, they can be presented as a category, a multicategory and a concategory. 
Mathematically, the difference between these presentations is almost trivial. 
But in terms of programming syntax, the difference between presentations has a big effect, as can be seen from the difference between the complex translation in Section~\ref{sec:simplediff} and the simple translations in Section~\ref{sec:unf}, which are much more readily verified.

\subsection{Denotational semantics} 
\label{sub:Denotational semantics source and Target}

A multicategory generalizes a category by allowing multimorphisms, that is, morphisms from a list of objects to an object.
Most categorical structures from category theory can be phrased similarly in multicategories.
It is standard to give a denotational semantics of a first-order language in a Cartesian category, and alternatively in a multicategory.

A term $x_1:A_1,\ldots,x_n:A_n\vdash e:A$ is interpreted in a multicategory as a morphism $\sem{e}:[\sem{A_1},\ldots,\sem{A_n}]\to \sem{A}$.
Substitution is interpreted as composition. We first consider a syntactic model for a language, which consists of a 
free multicategory on some base types and primitives. Our source and target languages induce syntactic multicategories as follows.

\begin{definition}[Syntactic multicategory for Source]
    Let $\SynSource$ be the multicategory whose objects are types of Source, and where a morphism 
    $[A_1,\ldots,A_n]\to A$ is a term $x_1:A_1,\ldots,x_n:A_n\vdash e:A$ of Source modulo the $\eta\beta$-laws.
    Composition is by substitution.
\end{definition}

We similarly define $\SynTarget$, the syntactic multicategory on the target language.

$\SynSource$ satisfies the following universal property: 
for every Cartesian multicategory $\catC$,
and every object $F(\reals)\in\catC$, morphisms $F(\underline{c})\in\catC(1;F(\reals))$, 
$F(op1)\in\catC(F(\reals);F(\reals))$, $F(op2)\in\catC(F(\reals),F(\reals);F(\reals))$, 
there is a unique multifunctor $F:\SynSource\to\catC$ respecting the interpretation and preserving all the categorical structure. 

This allows us to give a simple semantics of Source 
in the multicategory of Cartesian spaces and smooth maps between them. 

\begin{definition}[$\CartSp$]
Let $\CartSp$ be the Cartesian multicategory whose objects are Euclidean spaces
and whose morphisms $[A_1,\ldots,A_n]\to B$ are smooth functions $A_1\times\ldots A_n \to B$.

We interpret the source language in $\CartSp$ as follows. A context $\Gamma=\{x_1:A_1,\ldots,x_n:A_n\}$ is interpreted as the product $\prod_{1\leq i \leq n}\sem{A_i}$.
Well typed terms $\Gamma\vdash$!e: !$A$ are interpreted as functions $\sem{\Gamma}\to\sem{A}$.

    \begin{tabular}{r c l}
    $\seml \reals \semr$ & $\defeq$& $\RR$ \\
    $\seml$!T1xT2!$\semr$ & $\defeq$& $\seml$!T1!$\semr \times\seml$!T2!$\semr$ \\
    $\seml \Array{\reals}{n}\semr$ & $\defeq$ & $\prod_{1\leq i \leq n} \RR$ \\
    \end{tabular}
    \begin{tabular}{r c l} 
        !$\seml$op1$\semr$! & $\defeq$ & op1: $\reals\to\reals$ \\
        !$\seml$op2$\semr$! & $\defeq$ & op2: $\reals\times \reals\to\reals$ \\
        !$\seml$c$\semr$! & $\defeq$ & $c\in\RR$ \\
    \end{tabular}

Variables are interpreted as projections $\pi_i$, let binding as composition in the multicategory. 
Pairs are interpreted using the Cartesian structure of the multicategory.
This interpretation extends to !map2 (x,y.e$_1$) e$_2$ e$_3$!. It is given by
!<$\Delta_n,$,$\seml$ e$_2\semr$,$\seml$ e$_2\semr$>;swap;$(\seml$ e$_1\semr)^{\times n}$!
where $\Delta_n$ is the $n$-copy map !<id,$\ldots$,id>!, and !swap! a permutation.
Similarly, the semantics for !reduce (x,y.e$_1$) e$_2$ e$_3$! is then given by first !<$\Delta_n,$,$\seml$ e$_2\semr$,$\seml$ e$_2\semr$>!, followed by a permutation $\Gamma^n \times \RR \times \RR^n \to \reals \times (\reals \times \Gamma)^n$. 
Finally, we apply $\seml$ e$_1\semr \times id$ $n$ times, where the identity is of the obvious type at each stage.
\end{definition}

We can similarly interpret the target language in a multicategory of smooth-like spaces and functions. 
However, Target is higher-order and $\CartSp$ is not Cartesian Closed. 
Instead, we can interpret Target in the category of Diffeological spaces, 
as in \cite{huot2020correctness}. 
Diffeological spaces (\cite{iglesias2013diffeology}) are a conservative extension of $\CartSp$. 
The key idea will be that a higher-order function is called smooth if it sends smooth functions to smooth functions, meaning that we can never use it
to build first-order functions that are not smooth.

\begin{definition}
	A \emph{diffeological space} $(X,\plots{X})$ consists of a set $X$ together with, for each $n$ and each open subset $U$ of $\RR^n$,  a set $\plots{X}^U\subseteq [U\to X]$ of functions, called \emph{plots}, such that
	\begin{itemize}
	 	\item all constant functions are plots;
	 	\item if $f:V\to U$ is a smooth function and $p\in\plots{X}^U$, then $f;p\in\plots{X}^V$;
     \item if $\seq[i\in I]{p_i\in\plots{X}^{U_i}}$ is a compatible family of plots $(x\in U_i\cap U_j\Rightarrow p_i(x)=p_j(x))$
     and $\seq[i\in I]{U_i}$ covers $U$,
     then the gluing $p:U\to X:x\in U_i\mapsto p_i(x)$ is a plot.
	 \end{itemize} 
\end{definition}
We call a function $f:X\to Y$ between diffeological spaces \emph{smooth} if, for all plots
$p\in\plots{X}^U$, we have that $p;f\in \plots{Y}^U$. We write $\Diff(X,Y)$ for the set of smooth maps from $X$ to $Y$. 
Smooth functions compose, and so we have a category $\Diff$ of diffeological spaces and smooth functions.

A diffeological space is thus a set equipped with structure.
Many constructions of sets carry over straightforwardly to diffeological spaces.
For instance, given a family $\seq[i\in I]{X_i}$ of diffeological spaces,
we can equip the product $\prod_{i\in I}X_i$ of sets with the
\emph{product diffeology} in which $U$-plots are precisely the functions
of the form $\seq[i\in I]{p_i}$ for $p_i\in\plots{X_i}^U$.  
Cartesian spaces $\RR^n$ can be given the structure of a diffeological space by taking all the
smooth functions $U\to \RR^n$ as $\plots{\RR^n}^V$. We can equip the set $\Diff(X,Y)$ of smooth functions between diffeological spaces with the \emph{functional diffeology}
in which $U$-plots consist of functions $f:U\to \Diff(X,Y)$ such that 
$(u,x)\mapsto f(u)(x)$ is an element of $\Diff(U\times X, Y)$.
We can thus interpret function types !$\seml$A -> B$\semr$! = $\Diff$($\seml$!A!$\semr$,$\seml$!B!$\semr$).

\subsection{Semantics for UNF languages with concategories} 
\label{sub:Semantics for UNF using concategories}

One main reason for introducing UNF is to have a better handle over
the computation flow of the term, in the same vein as the ANF or CPS transformations.
A convenient categorical setting for this is to use string diagrams.
To better fit the standard denotational semantics of languages, 
we use concategories instead. Similar to the settings of categories and multicategories, 
most categorical constructions used for the semantics of functional languages 
have equivalent in concategories.

In particular, one can form a syntactic concategory on some base types and primitives, 
which satisfies a similar universal property as syntactic multicategories.
Two particular examples are of interest to us, as they will allow us to interpret Source UNF and Target UNF
but they will also help us explain the $\UNFSymbol$ and $UNF^{-1}$ transformations.

\begin{definition}[$\ConcatS$]
Let $\ConcatS$ be the syntactic Cartesian concategory whose types are those of Source and 
with primitives given by $op1: \reals \to \reals$, $op2: \reals,\reals \to \reals$, 
$map2_{x,y.e}: \Array{\reals}{n},\Array{\reals}{n}\to \Array{\reals}{n}$ and $reduce_{x,y.e1;e2}:\Array{\reals}{n}\to \reals$.
\end{definition}

One may notice that the syntax of $\ConcatS$ is somewhere in between the syntax of Source and of Source UNF.
Given two morphisms !e$_1$:T$_1$->T$_2$! and !e$_2$:T$_3$->T$_4$!, we denote by !e$_1 \concatcomp$ e$_2$:T$_1$,T$_3$->T$_2$,T$_4$! 
their parallel composition.
We denote by $u_T$ the unique morphism from !T! to the terminal object $[]$.
$pair_{T;A\times B}$ is the canonical isomorphism pairing !A,B! into !A$\times$B! in the context !T!.
We can interpret Source UNF in $\ConcatS$ as follows.

\begin{tabular}{l c l}
   $\seml$!var!$_{T;i} \semr$ &=& $<id_\Gamma,\pi_i>$ \\
   $\seml$!op!$_{T;n} \semr$ &=& $id_\Gamma\concatcomp op_n$\\
   $\seml$!pair!$_{T;A\times B} \semr$ &=& $pair_{T;A\times B}$ \\
   $\seml$!proj!$_{T1;T2;T3} \semr$ &=& $id_{T1}\concatcomp u_{T2}\concatcomp id_{T3}$
    \end{tabular}
   \begin{tabular}{l c l}
   $\seml$!e1!$\comp$!e2!$\semr$  &=& $\seml$!e1!$\semr \comp \seml$!e2!$\semr$ \\
   $\seml$!e1!$\pcomp$!e2!$\semr$ &=& $\seml$!e1!$\semr \comp (\seml$!e2!$\semr\concatcomp id )\comp swap$ \\
   $\seml$!map2!$_{T;x,y.e}\semr$  &=& $id_\Gamma\concatcomp map2_{x,y.e}$ \\
   $\seml$!reduce!$_{T;x,y.e;e}\semr$ &=& $id_\Gamma\concatcomp reduce_{x,y.e1;e2}$
\end{tabular}

Similarly, we introduce a second concategory for the Target part of our transformations.

\begin{definition}[$\ConcatT$]
Let $\ConcatT$ be the syntactic Cartesian concategory whose types are those of Target and 
whose primitives are given by 
$op1: \reals \to \reals$, 
$op2: \reals,\reals \to \reals$, 
$map2_{x,y.e}: \Array{\reals}{n},\Array{\reals}{n}\to \Array{\reals}{n}$, 
$reduce_{x,y.e1;e2}: \Array{\reals}{n} \to \Array{\reals}{n}$, 
$J^Top1: \reals \to \reals$, 
$J^Top2: \reals \to \reals, \reals$, 
$J^Tmap2_{x,y.e}: \Array{\reals}{n} \to \Array{\reals}{n}, \Array{\reals}{n}$, 
$J^Treduce_{x,y.e1;e2}: \Array{\reals}{n} \to \Array{\reals}{n}$ 
\end{definition}

We can interpret Target UNF in $\ConcatT$ as follows. 
The part common with Source UNF is interpreted in the same way as for the case of Source UNF.

\begin{tabular}{l c l}
    $\seml$!J!$^T$!var!$_{T;i} \semr$ &=& $\pi_T \widehat{+}(0,\pi_A,0)$ \\
    $\seml$!J!$^T$!map2!$_{T;x,y.e}\semr$  &=& $id_T\concatcomp map2_{T;x,y.e}$ \\
    $\seml$!e$_1$!$\icomp$!e$_2$!$\semr$ &=& $\Lambda$($id_\Gamma\concatcomp\seml$!e$_2$!$\semr$; $\Lambda^{-1}$($\seml$!e$_1$!$\semr$)) \\
    $\seml$!J!$^T$!reduce!$_{T;x,y.e_1;e_2}\semr$  &=& $id_T\concatcomp reduce_{T;x,y.e_1;e_2}$ \\
    $\seml$!J!$^T$!proj!$_{T1;T2;T3} \semr$ &=& $id_{T_1}\concatcomp\widehat{0}_{T_2}\concatcomp id_{T_3}$ 
    \end{tabular}
    \begin{tabular}{l c l}
    $\seml$!J!$^T$!op!$_{T;n} \semr$ &=& $id_T\concatcomp J^Top_n$ \\
    $\seml$!<e$_1$, e$_2$>!$\semr$  &=& !<!$\seml$!e$_1$!$\semr$, $\seml$!e$_2$!$\semr$!>! \\
    $\seml$!J!$^T$!pair!$_{T;A\times B} \semr$ &=& $id_T\concatcomp$!<$\pi_1,\pi_2$>! \\
 \end{tabular}

 \noindent
 where !T,A! $\vdash$ (0,$\pi_A$,0): !T! and the projection $\pi_A$ lands in the $i$-th element of the list !T!.

\subsection{Semantics of $\UNFSymbol$ transformations} 
\label{sub:Semantics for UNF transformations}

We interpret the source language in a new multicategory, 
whose morphisms are particular morphisms of $\ConcatS$.
As Source UNF is itself interpreted in $\ConcatS$, 
this gives us a way to compare terms in Source with terms in Source UNF.
This comparison of morphisms of $\ConcatS$ gives the $\UNFSymbol$ transformation.

\begin{definition}[Mutlicategory from concat]
    We define $\CSource$ to be the multicategory with the same objects as Source and\\ 
    $\CSource([A_1,\ldots,A_n],B)=\{f\in\ConcatS([A_1,\ldots,A_n],[A_1,\ldots,A_n,B], \forall i.f;\pi_{A_i}=id_{A_i})$\\
    The composition of $f_i:\underline{A}_i\to B_i$ with $g:\underline{B}\to C$ is given 
    $f_1\concatcomp\ldots\concatcomp f_n;(u_{\underline{B}_1}\concatcomp A_1 \concatcomp \ldots \concatcomp u_{\underline{B}_n}\concatcomp A_n);g$.
    In other words each term $f_i$ forgets about its output $\underline{B}_i$, and then we use the composition in the concategory.
\end{definition}

We can interpret Source in $\CSource$ as follows. The functor is an identity on types. 
On morphisms, we interpret them as morphisms in the concategory as follows.
The terminal map is interpreted as the identity. The operators !op1,op2! by themselves.
Crucially, the semantics of a !let! is simple composition in the concategory. 
The semantics of a variable is the pairing of the identity with a projection.
The operation !map2 (x,y.e)! is interpreted as the pairing of the identity and itself, and !reduce (x,y.e$_1$) $e_2$! as itself.

Interpreting Source in $\CSource$ allows us to see terms of Source as morphisms in the concategory $\ConcatS$, 
and to compare them to terms of Source UNF which are already interpreted in $\CSource$.
The following proposition can be shown by induction on the structure of the terms.

\begin{proposition}[construction above gives UNF]
    Let !$\Gamma \vdash$ e:A! be a term in Source. 
    Seen as morphisms in $\CSource$, !$\seml$e$\semr$=$\seml\UNFSymbol$(e)$\semr$!.
\end{proposition}

Dually, we form a concategory from the syntactic multicategory for Target.
Then we use the universal property of $\ConcatT$ to construct a functor from $\ConcatT$ to this concategory.
This allows us to compare the terms of Target UNF and the terms of Target, and $\UNFSymbol^{-1}$ arises in this way.

\begin{definition}[Concat from multicat]
    A multicategory $\catC$ naturally defines a concategory with the same objects as $\catC$ and
    with morphisms $\underline{A}\to [B_1,\ldots,B_n]$ being $n$ morphisms $\underline{A}\to B_i$ of $\catC$. 
\end{definition}

We thus consider the concategory $\CTarget$ from the syntactic multicategory for Target.
We can interpret $\ConcatT$ in $\CTarget$. 
The functor is identity-on-objects, sends operations !op1,op2! to themselves. 
It sends Jacobian operations to terms of Target as given by $\UNFSymbol^{-1}$ in Figure~\ref{fig:unf_to_target}.

This interpretation is in essence $\UNFSymbol^{-1}$. 
The difference is that to preserve typing, the semantic $\UNFSymbol^{-1}$
sends a non-Jacobian primitive to a tuple, as in Example \ref{exm:unf}. 
This is highly inefficient, and the syntactic $\UNFSymbol^{-1}$ 
additionally projects to the last element, where the non-trivial information of the term is.

What remains to explain now is the reverse mode transformation between Source UNF and Target UNF.
We construct a functor $\Dsynrevsymbol_{\rho}:\ConcatS\to\ConcatT$. 
This functor computes reverse-mode derivatives. 
Because terms of Source UNF are interpreted in $\ConcatS$, we observe the effect of $\Dsynrevsymbol_{\rho}$ on them
and show that it matches the syntactic $\Dsynrevsymbol_{\rho}$ 
from Section~\ref{sub:Simple reverse mode transformation}.

\begin{definition}[$\Dsynrevsymbol_{\rho}$ as a lax functor]
    We define $\Dsynrevsymbol_{\rho}:\ConcatS\to\ConcatT$ as follows.\\
    !$\Dsynrevsymbol_{\rho}$([A$_1$,$\ldots$,A$_n$])=[A$_1$,$\ldots$,A$_n$,A$_1\times\ldots\times$A$_n\to \rho$]!.
    !$\Dsynrevsymbol_{\rho}$(op1)=<op1,$\pi_{last}\circ$ J$^T$op1>!,\\ 
    !$\Dsynrevsymbol_{\rho}$(op2)=<op2,$\pi_{last}\circ$ J$^T$op2>!,\\
    !$\Dsynrevsymbol_{\rho}$(map2 (x,y.e))=<map2 (x,y.e),$\pi_{last}\circ$ J$^T$map2 (x,y.e)>!,\\
    !$\Dsynrevsymbol_{\rho}$(reduce$_{x,y.e_1,e_2}$)=<reduce$_{x,y.e_1,e_2}$,$\pi_{last}\circ$ J$^T$reduce$_{x,y.e_1,e_2}$>!, 
    where $\pi_{last}$ is the projection to the last element (the continuation variable).
    It naturally extends to sequential composition. 
    It does not automatically extend to a multi-functor 
    \\!$\Dsynrevsymbol_{\rho}$([A$_1$,$\ldots$,A$_k$]$\concatcomp$[A$_{k+1}$,$\ldots$,A$_n$]) $\neq$ $\Dsynrevsymbol_{\rho}$([A$_1$,$\ldots$,A$_k$])$\concatcomp\Dsynrevsymbol_{\rho}$([A$_{k+1}$,$\ldots$,A$_n$])!.\\
    Still, there is a map !$\Dsynrevsymbol_{\rho}$([A$_1$,$\ldots$,A$_k$])$\concatcomp \Dsynrevsymbol_{\rho}$([A$_{k+1}$,$\ldots$,A$_n$])$\to \Dsynrevsymbol_{\rho}$([A$_1$,$\ldots$,A$_k$]$\concatcomp$[A$_{k+1}$,$\ldots$,A$_n$])!.
    Internally, as a lambda term, it is given by \\
    $(x_1,\ldots,x_k,Y_1,(x_{k+1},\ldots,x_n,Y_s)\mapsto (x_1,\ldots,x_n,\lambda(y_1,\ldots,y_n)\to Y_1(y_1,\ldots,y_k)\widehat{+}Y_2(y_{k+1},\ldots,y_n)$.
    Here, $\widehat{+}$ is the reverse-derivative of the copy map, which is known to be fanout 
    . 
    It is not surprising to see it appear as $\Dsynrevsymbol_{\rho}$ is a semantic functor, and it does not need to be efficient. 
\end{definition}

The design of the syntactic $\Dsynrevsymbol_{\rho}$ on the UNF language is inspired by the semantic one between concategories, 
and it's routine to check that they match. 

\begin{proposition}[semantic of syntactic D matches D lax functor]
    Given a well typed term !T$_1$ $\vdash$ e: T$_2$! in Source UNF, we have
    \begin{center}
        $\seml\Dsynrevsymbol_{\rho}$!e!$\semr$ = $\Dsynrevsymbol_{\rho}\seml$!e!$\semr$
    \end{center}
\end{proposition}

In summary, we have the following picture (not a proper categorical diagram):

\[
\begin{tikzcd}
    SourceUNF \ar[r,"\sem{-}"] & Concat1 \ar[r,"\Dsynrevsymbol_{\rho}"] & Concat2 \ar[r,"\UNFSymbol^{-1}"] & C_{Target} \\
    Source \ar[r,"\UNFSymbol"] & \CSource \ar[u, hook] & TargetUNF \ar[u,"\sem{-}"] & Target \ar[u, hook] 
\end{tikzcd}
\]

\subsection{Correctness theorem} 
\label{sub:Correctness theorem}

First, we start from the correctness of the syntactic $\Dsynrevsymbol_{\rho}$: Source UNF $\to$ Target UNF, which is easy to establish, 
and then propagate this information to Source and Target via the $\UNFSymbol$ and $\UNFSymbol^{-1}$ transformations. 
The semantics brackets $\sem{-}$ in this section are in diffeological spaces.

\begin{proposition}[Correctness $\Dsynrevsymbol_{\rho}$]
    For every term !$\reals^{\times n} \vdash$ e: $\reals^{\times n+1}$! in Source UNF,
    \begin{center}
        $\pi_2 \seml \Dsynrevsymbol_{\rho}$!e!$\semr(x_1,\ldots x_n,Id_{\RR^n})$=$J^T_{(x_1,\ldots x_n)}\seml$!e!$\semr$
    \end{center} 
\end{proposition}

This is routinely proved by induction, as the language is first-order. 
This uses the fact that for every primitive constant $A$, $\sem{J^TA}=J^T\sem{A}$.

Recall that the intuition from Source UNF is that it consists of terms of Source that also return their context.
From there, the intuition for the Jacobian of a primitive in Target UNF is that it should be the Jacobian of
the corresponding term in Target. This is easily checked for scalar operations. 
For the non-trivial cases, we have 

\begin{proposition}
$\UNFSymbol^{-1}$ preserves the semantics of Jacobians of !map2! and !reduce!.
    \begin{center}
\begin{tabular}{r c l}
    !$\seml\UNFSymbol^{-1}(J^T$map2$_{x,y.e})\semr$! &=& !$J^T\seml\UNFSymbol^{-1}($map2$_{x,y.e})\semr$!\\
    !$\seml\UNFSymbol^{-1}(J^T$reduce$_{x,y.e_1;e_2})\semr$! &=& !$J^T\seml\UNFSymbol^{-1}($reduce$_{x,y.e_1;e_2})\semr$!
\end{tabular}
\end{center}
\end{proposition}

This is proved in the supplementary material.

From this, we now deduce that the composite transformation $\UNFSymbol$, $\Dsynrevsymbol_{\rho}$, $\UNFSymbol^{-1}$ is correct
in the sense that it produces a term that computes the gradient of the original term.

\begin{proposition}
    If !x$_1:\reals$,$\ldots$,x$_n:\reals$ $\vdash$ e: $\reals$! then \\
    $\seml \pi_2 \UNFSymbol^{-1}(\Dsynrevsymbol_{\rho}(\UNFSymbol($!e!$)))\semr(x_1,\ldots x_n,Id_{\RR^n})$=$J^T_{(x_1,\ldots x_n)}$!<!$Id_{\RR^n}$,$\seml$!e!$\semr$!>!.
\end{proposition}

By inspecting what that composition of transformations does on the terms of Source, 
we show that this indeed computes the same as the transformation 
$\directD{\rho}{\Gamma}{Y}$ from Section~\ref{sub:Macro for pure reverse mode transformation}. 
This should not come as a surprise because the design of $\directD{\rho}{\Gamma}{Y}$ was in fact guided via 
this decomposition and intermediate representation.

\begin{proposition}
    $\seml\UNFSymbol^{-1}(\Dsynrevsymbol_{\rho}(\UNFSymbol($!e!$)))\semr$ = $\seml\directD{\rho}{\Gamma}{Y}($!e!$)\semr$. 
\end{proposition}

Combining the previous propositions, we have shown that $\directD{\rho}{\Gamma}{Y}$ is correct.

\begin{theorem}
    For every $\Gamma \vdash$ !e:!$\reals$, we have 
    $\sem{\grad_\Gamma e}= \grad_\Gamma\sem{e}$.
\end{theorem} 
\section{Beyond the Source language}
\label{sec:generalization}
In this section, we generalize our language by relaxing the imposed restrictions.
First, we allow free variables for the function argument of !reduce! in Section~\ref{sub:lift_reduce}. 
Then we show how to support conditionals in Section~\ref{sub:lift_conditional}.
Finally, we provide a recipe for adding more constructs to the language in Section~\ref{sub:lift_recipe}. Further generalizations for more array operations, non-smooth scalar operators, and general array support can be found in the supplementary materials.

\subsection{Lifting the restriction on reduce}
\label{sub:lift_reduce}

Assume we allow $\Gamma$, !x!: $\reals$, !y!: $\reals$ $\vdash$ !e!: $\reals$ to be the function argument in
!reduce (x,y.e) v A!. We need to add a term depending on $\grad_{\{xi\}}e$ to !yi! in the continuation.
As writing the derivative becomes quite cumbersome, we use a more compact notation using arrays of tuples. 
Similarly to the monoid $\widehat{+},0_\Gamma$ defined in Section~\ref{sub:Macro for pure reverse mode transformation},
we use $\widehat{\times},1_\Gamma$ for obvious extension of the monoid $(\reals,\times,1)$. 
The modified reverse-mode transformation is shown in Figure~\ref{fig:lift_reduce}.

\begin{figure}

\begin{small}
\addtolength{\tabcolsep}{-2pt}
\begin{tabular}{|r c l|}
\hline
$\directD{\rho}{\Gamma}{Y}$(!reduce (x,y.e1) e$_2$ e$_3$!) 
&=& $\ldots$ same as Figure~\ref{fig:direct_diff_macro} up until !A_$3$! \\
&& !let B$_0$=map2 (a,b ->(!$\grad_\Gamma$!e)[a/x,b/y] A$_0$ A) in!\\
&& !let B$_1$=scanr $1_\Gamma$ $\widehat{\times}$ B$_0$ in! \\
&& !let B$_2$=reduce $\widehat{+}$ $0_\Gamma$ B$_1$ in! \\
&& !<reduce (x,y.e1) y$_1$ A, fun ($\boldsymbol{x}$),z) ->! \\
&& !Y($\boldsymbol{x}\widehat{+}(z*$B$_2$),0,map2 (a,b. a*b*z) A$_2$ A$_3$)! \\ \hline
\end{tabular}
\end{small}
\vspace{-0.4cm}
\caption{The reverse-mode AD transformation of \codekw{reduce} without restrictions on its function argument.}
\vspace{-0.4cm}
\label{fig:lift_reduce}
\end{figure}

\subsection{Conditionals} 
\label{sub:lift_conditional}

\smartpara{Non-Solution 1}
Adding conditionals to the source language can be done easily, for instance via defining

\begin{tabular}{r c l}
$\directD{\rho}{\Gamma}{Y}$(!if e$_1$ then e$_2$ else e$_3$!) &=& !if e$_1$ then $\directD{\rho}{\Gamma}{Y}$(e$_2$) else $\directD{\rho}{\Gamma}{Y}$(e$_3$)! 
\end{tabular}

The problem is that this could break the complexity of reverse mode because of the non-linear usage of $Y$, and makes everything harder to optimize.

\smartpara{Non-Solution 2}
A slightly better option would be to define 

\begin{tabular}{l}
    $\directD{\rho}{\Gamma}{Y}$(!if e$_1$ then e$_2$ else e$_3$!) = \\
    !let b=e$_1$ in <if b then e$_2$ else e$_3$, fun (x$_1$,$\ldots$,x$_n$,z) ->! \\
    !Y(b*!$\directD{\rho}{\Gamma}{Y}$!(e$_2$)(x$_1$,$\ldots$,x$_n$,z)+(1-b)*!$\directD{\rho}{\Gamma}{Y}$!(e$_3$)(x$_1$,$\ldots$,x$_n$,z))>!
\end{tabular}

Now both derivatives of !e2! and !e3! are put together, and this might unlock some optimizations, 
but there is still a non-linear usage of the continuation variable !Y!.

\smartpara{Solution}
If we know we want to compute the whole gradient of the expression, we can define the translation as follows:

\begin{tabular}{l}
    !$\directD{\rho}{\Gamma}{Y}$(if e$_1$ then e$_2$ else e$_3$) = let b=e$_1$ in <if b then e$_2$ else e$_3$,! \\
    !fun (x$_1$,$\ldots$,x$_n$,z) -> Y((x$_1$,$\ldots$,x$_n$)!$\widehat{+}$($\grad_\Gamma$!(e$_2$)*b!$\widehat{+}\grad_\Gamma$!(e$_3$)*(1-b))*z)>!\\
\end{tabular}

This time, there is a linear usage of the continuation variable !Y!. We note that adding conditionals does not break differentiability as long as there are no non-trivial primitives $\reals\to\BB$.
Non-smooth functions such as the Rectified Linear Unit (ReLU) in machine learning, which can be defined as $ReLU(x)\defeq max(0,x)$ 
requires a non-smooth primitive such as $>0: \reals\to\BB$. In several AD systems such as TensorFlow, 
the condition is evaluated before differentiation is applied, and so conditionals are never directly differentiated.

\subsection{Recipe for Adding More Constructs}
\label{sub:lift_recipe}
To add more constructs to our language, 
there is no need to go through UNF; most of the backend of our work can be used as a black box.
Non-smooth variables (typically booleans or integers) should be considered as external to the language, similarly to the way we treat indices $n$ of arrays.

\smartpara{Reverse AD}
If the operator has type $A\to B$, then one should provide its transpose Jacobian, a term $B\to A$. 
Such operators can often be unrolled to first-order programs. 

\smartpara{Correctness}
To check the correctness of the given Jacobian, it suffices to 
check that the semantics of the transpose Jacobian matches the Jacobian of the unrolled program. 

\smartpara{Complexity}
To ensure the complexity guarantee of the whole transformation, one needs to check that there is a linear usage of the continuation variable $Y$ and that the cost of the 
proposed Jacobian is at most $k$ times the complexity of the operator. 

\section{Discussion and future work} 
\label{sec:discussion_and_future_work}

\noindent \textbf{Design space}
First, there is a tradeoff to reach between a general expressive language and a domain specific one. The latter usually has more static information and a specific representation that lends itself to better optimizations.
Then, many optimizations performed on AD implementations consist in hand-crafted derivatives for useful operations like matrix-matrix multiplication, dot-product, etc.
They don't seem to arise from theoretical justifications, are error-prone, and can hardly fit with more general optimizations.
This makes these systems harder to prove correct. The problem is thus to ensure provable correctness and pureness, while not compromising on efficiency. In addition, we would like something easily, provably efficient. 
A real-world implementation based on our work would of course optimize further, potentially using hand-crafted operations as well, but it would be based on solid grounds.

\noindent \textbf{Higher-order functions}
Some recent work \cite{vakar2020reverse,sherman2021} present reverse-mode in a higher-order language.
\cite{vakar2020reverse} uses categorical semantics to show correctness of the reverse-mode transformation 
and \cite{sherman2021} uses sophisticated higher-order primitives such as root finding, max, argmax or integral. 
Their work focuses on computable reals, which is hard to compare in terms of efficiency with our more standard approach of AD.
It is however quite difficult to prove a complexity result in the higher-order setting. 
In addition, standard techniques of defunctionalization struggle when higher-order is combined with recursion.
Our reverse-mode transformation does not support recursion, 
so it is currently always possible to partially evaluate a higher-order program to our Source language, seen as intermediate representation,
then perform our reverse-mode transformation.

\noindent \textbf{Array primitives}
We focused on giving reverse derivatives for a small set of array primitives. 
Other common primitives include filter, flatten, gather. 
These functions can easily be added to our Source language once we provide a reverse derivative for them.
This is reminiscent of hand-crafted derivatives present in large AD frameworks (usually hundreds), 
and these can already be added in our language. 
As shown in \ref{sub:lift_recipe}, one mostly only needs to make sure that the provided transpose Jacobian is correct.





\section{Related Work} 
\label{sec:related_work}

\noindent \textbf{Correctness of AD in functional languages.}
Several recent works \cite{huot2020correctness,vakar2020reverse,vakar2020denotational,brunel2019backpropagation,barthe2020versatility,mazza2021automatic,lee2020correctness,abadi-plotkin2020} have focused on correctness of AD in a purely functional setting, 
often leaving efficiency on the side, especially for reverse-mode differentiation. 
We see our work as a complement and a first bridge between these works 
and more practical considerations of efficiency, 
which often require a lot more care than is acknowledged in more theoretical works.


\noindent \textbf{Usage of iteration mechanics.}
An immense effort in machine learning for the past decade has been in finding
good architectures, to limit computational costs, 
avoid vanishing and exploding gradients, 
and have better building blocks for large complicated systems than traditional layers of a neural network.
Different approaches such as Dynamic neural networks \cite{jin2017manipulability,wu2016deep}, 
Recursive NN \cite{socher2011parsing,biancofiore2017recursive}, 
Reccurent NN \cite{bahdanau2014neural,luong2015effective}, 
Tree LSTM \cite{tai2015improved,chen2016enhanced}, 
Dynamic Recursive NN \cite{guo2019dynamic}, 
Top-down Tree LSTM \cite{zhang2015top}, and
Recursion in DNN \cite{jeong2018improving}
have found that recursive data structures such as trees are good candidates.
We have emphasized here on differentiating fold-based recursion on arrays for efficiency, 
but one should be able to adapt this to any algebraic data type. 
It will be interesting to see if and how we recover efficient purely functional backpropagation (as opposed to the imperative version of~\cite{lantern_icfp}) on the proposed architectures, 
which is usually derived by hand and one main goal of these papers.

\noindent \textbf{Array Languages and AD.}
Given the enormous computation needs for state-of-the-art large scale machine learning applications, 
which require extremely efficient tensor computations and automatic differentiation for backpropagation, 
combining array languages and automatic differentiation (tensor calculus) in the best-fitted 
intermediate representation for optimizations is of key interest and active research~\cite{bernstein2020differentiating,laue2018computing,laue2020simple}. 
Advanced array programming is considered an orthogonal problem to AD, and we focused our work on the differentiation aspect.

\noindent \textbf{Comparison to other recent papers.}
Table~\ref{tbl:relwork} does not reflect all the aspects of AD. For instance, \cite{lee2020correctness} studies in more detail non differentiability, 
and \cite{sherman2021} the differentiability of highly non-trivial higher-order and partial functions. 
Our work is somewhat close in spirit to the idea of \cite{Elliott:2018:SEA:3243631.3236765} of compiling to categories.
The idea of using closures as back-propagators is receiving recent attention, as is highlighted in \cite{vytiniotis2019differentiable,lantern_icfp}.
These ideas are used in Julia Zygote \cite{innes2019zygote}, Swift AD \cite{wei2018first}, and recently in \cite{paszke2021getting}.
These seem closer to using control mechanisms than having purely functional reverse-derivatives.
Other aspects of AD are discussed in recent surveys \cite{van2018automatic,baydin2017automatic}. \cite{krawiec2022provably} show how to do efficient reverse-mode AD for a higher-order purely functional language, but at the cost of requiring a monadic translation.

\section{Conclusion}
\label{sec:conclusion}

We introduced a transformation on programs to compute provably efficient (\S\ref{sec:complexity}) 
gradients via reverse derivatives in a purely functional way (\S\ref{sub:Macro for pure reverse mode transformation})
on a simple yet expressive language with functions on arrays (\S\ref{sub:sourcelang}), 
combined with standard functional optimizations (\S\ref{fig:optim}).  
We introduced a novel intermediate representation, Unary Normal Form (\S\ref{sec:unf}) 
to decompose our translation into simpler ones.
We gave denotational semantics to our languages (\S\ref{sec:correctness}), 
and we proved the correctness of the reverse-mode translation (\S\ref{sub:Correctness theorem}).
We showed (\S\ref{sec:generalization}) how to lift the restrictions that
we introduced on arrays and how to extend our approach with other constructs such as conditionals.

\begin{acks}
We have benefited from discussing this work with many people, including Younesse~Kaddar, Jesse~Sigal, Matthijs~V\'ak\'ar, Emmanuel Arrighi, Sam Staton and others. The first author is supported by a Royal Society University Research Fellowship.
The second author thanks Huawei for their support of the distributed
data management and processing laboratory at the University of Edinburgh.
\end{acks}

\bibliography{refs}

\appendix
 \section{Appendix}

 \subsection{Operational semantics}

In Figure~\ref{fig:op_semantics_target} a small step call-by-value operational semantics for the language. 
The evaluation contexts aregiven in Figure~\ref{fig:ev_contexts}. 

\begin{figure*}[t]
    \begin{tabular}{|l c l|}
        \hline
        \multicolumn{3}{|l|}{Evaluation contexts} \\
        $E$ & \mbox{::=} & 
        [] 
        $\mid$ !let x = E in e! 
        $\mid$ !<E, e>!
        $\mid$ !<v, E>!
        $\mid$ $\pi_i$!(E)!
        $\mid$ !E op2 e!
        $\mid$ !v op2 e!
        $\mid$ !op1 E! \\
        && $\mid$ !map (x.e) E!
        $\mid$ !map2 (x,y.e) E e!
        $\mid$ !map2 (x,y.e) v E! \\
        && $\mid$ !foldl (x,y.e) E e!
        $\mid$ !foldl (x,y.e) v e! \\
        && $\mid$ !reduce (x,y.e) E e! 
        $\mid$ !reduce (x,y.e) v e! \\
        && $\mid$ !scanl (x,y.e) E e!
        $\mid$ !scanl (x,y.e) v e! \\
        && $\mid$ !scanr (x,y.e) E e!
        $\mid$ !scanr (x,y.e) v e! \\
        && $\mid$ !shift1L E! $\mid$ !shift1R E! \\
        && $\mid$ !E(e$\ldots$e)!
        $\mid$ !e(v$\ldots$vEe$\ldots$e)!
        $\mid$ !if E then e else e! \\
        && $\mid$ ![v,$\ldots$,v,E,e$\ldots$,e]!
        \\ \hline
        \multicolumn{3}{|l|}{Values} \\ 
        !v! & \mbox{::=} & 
        \cnst{}  
        $\mid$ !<v, v>!
        $\mid$ !true! 
        $\mid$ !false!
        $\mid$ !fun (x1,$\ldots$,xn) -> e!
        $\mid$ ![v,$\ldots$,v]! 
        \\ \hline
        \end{tabular}
    \caption{Evaluation contexts and values}
\label{fig:ev_contexts}    
\end{figure*}

\begin{figure*}[tb]
\begin{tabular}{|l c l|}
    \hline
    !op1! \cnst{} & \transto & \underline{op1(c)} \\ \hline
    \cnst{} !op2! \cnst{}' & \transto & \underline{op2(c,c')}\\ \hline
    !let x=v in e! & \transto & !e[v/x]!  \\ \hline
    $\pi_i$!<v$_1$, v$_2$>! & \transto & !vi!\\ \hline
    !(fun (x$_1$,$\ldots$,x$_n$) -> e)(v$_1$$\ldots$v$_n$)! & \transto & !e[v$_1$/x$_1$,$\ldots$,v$_n$/x$_n$]! \\ \hline
    !shift1L [v$_1$,$\ldots$,v$_{n+1}$]! & \transto & ![v$_2$,$\ldots$,v$_{n+1}$]! \\ \hline
    !shift1L []! & \transto & ![]! \\ \hline
    !shift1R [v$_1$,$\ldots$,v$_{n+1}$]! & \transto & ![v$_1$,$\ldots$,v$_n$]! \\ \hline
    !shift1R []! & \transto & ![]! \\ \hline
    !scanl (x,y.e) v []! & \transto & ![v]! \\\hline
    !scanl (x,y.e) v [v$_1$,$\ldots$,v$_n$]! & \transto & !v::(scanl (x,y.e) e[v/x,v$_1$/y]! \\ 
    && !           [v$_2$,$\ldots$,v$_n$])!\\ \hline
    !scanr (x,y.e) v []! & \transto & ![v]! \\ \hline
    !scanr (x,y.e) v [v$_1$,$\ldots$,v$_n$]! & \transto & !(scanl (x,y.e) e[v/x,v$_1$/y]! \\
    && !       [v$_2$,$\ldots$,v$_n$])::v! \\ \hline
    !reduce (x,y.e) v []! & \transto & !v! \\ \hline
    !reduce (x,y.e) v [v$_1$,$\ldots$,v$_n$]! & \transto  & !reduce (x,y.e) e[v/x,v$_1$/y] [v$_2$,$\ldots$,v$_n$])!\\ \hline
    !map2 (x,y.e) [v$_{11}$,$\ldots$,v$_{1n}$]! & \multirow{2}{*}{\transto} & ![e[v$_{11}$/x,v$_{21}$/y],$\ldots$,e[v$_{1n}$/x,v$_{2n}$/y]]! \\ 
    ![v$_{21}$,$\ldots$,v$_{2n}$]! && \\ \hline
    !if true then e$_1$ else e$_2$! & \transto  & !e$_1$! \\ \hline
    !if false then e$_1$ else e$_2$! & \transto  & !e$_2$! \\ \hline
    \end{tabular}
\vspace{-0.2cm}
\caption{Operational semantics of the source and target languages}
\vspace{-0.4cm}
\label{fig:op_semantics_target}
\end{figure*}

 \subsection{Reverse derivative of array operations}

 \label{sub:Reverse derivative of array operations}

We now prove that the reverse mode transformation is correct on array operations. 

 \begin{proposition}
    The reverse derivative of !map2! is correct.
\end{proposition}

\begin{proof}
Let us denote by !t! the term !map2 (x,y.e) A B!. Without loss of generality, assume !A!=$[a_1,\ldots,a_n]$ and !B!=$[b_1,\ldots,b_n]$.
By chosing !Z! to a hot vector at i, call it !Z$_i$!, we're back to showing the result for the term $\Gamma \vdash$!e[a$_i$/x,b$_i$/y]: $\reals$!.
In !$\directD{\rho}{\Gamma}{Y}$(t)(Z$_i$)!, the term !G = (map2 * Z (map2 (a,b.($\grad_\Gamma$e$_{1}$)[a/x.b/y]) A B))! 
reduces to $\frac{\partial e}{\partial x_i}$[a$_i$/x,b$_i$/y]. As !a$_i$,b$_i$! are independant of $x_i$, 
this term is equal to $\frac{\partial e[a_i/x,b_i/y]}{\partial x_i}$, as expected.
Similarly, the term !map2 * (map2 (a,b.(!$\grad_{\{x\}}$!e$_{1}$)[a/x,b/y]) A B) Z! reduces to $\frac{\partial e}{\partial x}$[a$_i$/x,b$_i$/y].
As !a$_i$! is a variable, and independant of !b$_i$!, this is equal to $\frac{\partial e[a_i/x,b_i/y]}{\partial a_i}$.
\end{proof}

\begin{lemma}
    if !op2! is an associative binary operation with unit $\Gamma \vdash$ !v:! $\reals$, then 
    $\frac{\partial op2(v,e)}{\partial y_1}\times\frac{\partial v}{\partial x_i}=0$ 
    and $\frac{\partial op2(e,v)}{\partial y_2}\times\frac{\partial v}{\partial x_i}=0$ for all $x_i$.
\end{lemma}

\begin{proof}
    For any $\Gamma \vdash$ !e:! $\reals$, we have !v op2 e!=!e!.
    Differentiating and using the chain rule we get 
    $$\frac{\partial op2(v,e)}{\partial y_1}\times\frac{\partial v}{\partial x_i}
    +\frac{\partial op2(v,e)}{\partial y_2}\times\frac{\partial e}{\partial x_i}
    = \frac{\partial e}{\partial x_i}$$
As $\frac{\partial e}{\partial x_i}$ is arbitrary, 
this shows that $\frac{\partial op2(v,e)}{\partial y_2}=1$ and $\frac{\partial op2(v,e)}{\partial y_1}\times\frac{\partial v}{\partial x_i}=0$.
Similarly for the other case.
\end{proof}

\begin{proposition}
    The reverse derivative of !reduce! is correct.
\end{proposition}

\begin{proof}
    The notation for the general case are cumbersome and non-insightful. 
    We will exemplify the proof on the case of an array of size 3. We use infix notation for the binary operation $e$.
    The output term unrolls to $v_3\defeq ((v ~e ~a_1)~ e~ a_2) ~e~ a_3$, where $v$ is the unit of $e$.
    By the lemma above we know that the derivative w.r.t $v$ is 0, and focus on the partial derivatives w.r.t $a_i$.
    Write $v_1\defeq(v~ e~ a_1)$, $v_2 \defeq v_1 ~e ~a_2$.
    By inspection we have
    \begin{itemize}
        \item $\frac{\partial v_3}{\partial a_3}=\frac{\partial e}{\partial x_2}(v_2,a_3)$
        \item $\frac{\partial v_3}{\partial a_2}=\frac{\partial v_3}{\partial x_1}\times\frac{\partial v_2}{\partial x_2}=\frac{\partial e}{\partial x_1}(v_2,a_3)\times \frac{\partial e}{\partial x_2}(v_1,a_2)$
        \item $\frac{\partial v_3}{\partial a_1}=\frac{\partial v_3}{\partial x_1}\times\frac{\partial v_2}{\partial x_2}\times\frac{\partial v_1}{\partial x_2}=\frac{\partial e}{\partial x_1}(v_2,a_3)\times \frac{\partial e}{\partial x_1}(v_1,a_2)\times\frac{\partial e}{\partial x_2}(v,a_1)$
    \end{itemize}
    We thus need the following intermediate results
    \begin{itemize}
        \item $A_0=[v,v_1,v_2]$=!shift1R(scanl v e A)!
        \item $A_1 =[\grad_{\{x_1\}}e(v_1,a_2),\grad_{\{x_1\}}e(v_2,a_3)]=$ !shift1L (map2 (a,b.$\grad_{\{x_1\}}$e(a/x,b/y)) $A_0$ A)!
        \item $A_2 =[\grad_{\{x_2\}}e(v,a_1),\grad_{\{x_2\}}e(v_1,a_2),\grad_{\{x_2\}}e(v_2,a_3)]=$ !map2 (a,b.$\grad_{\{x_2\}}$e(a/x,b/y)) $A_0$ A!
        \item $A_3 =[\frac{\partial e}{\partial x_1}(v_2,a_3)\times \frac{\partial e}{\partial x_1}(v_1,a_2),\frac{\partial e}{\partial x_1}(v_2,a_3),1]=$ !scanr 1 * $A_1$!
    \end{itemize}
    And we return $A_4=$ !map2 $A_2$ $A_3$!.
\end{proof}

\subsection{Adding more array operators}
\label{sub:lift_more_arr_op}

There are two main differences with fold left !foldl! compared to !reduce!. 
First, in !foldl (x,y.e) v A! the starting accumulation element !v! is not a unit for !(x,y.e)!,
so it will have non-trivial derivatives in general and we need to account for that.
Second, we will need a more general !scanl! which allows !(x,y.e)! as a function argument. 
In other words, we need the general scan left computing the intermediate values of a fold left.
This should be a different primitive but we will still call it !scanl! in this section.

Finally, the former point implies we need a few more array manipulations.
These can be elegantly dealt with by changing the semantics of !scanl! and !scanr!. 
Let's assume that they now return a pair of an array of size $n$ 
of the intermediate computations and the final result.
The reverse derivative of !foldl! is then as shown in Figure~\ref{fig:foldl_map_trans}.
In this figure, we have shown the translation of the !map! operator as well, which is very similar to !map2!.

\begin{figure}
\begin{center}
\begin{tabular}{|r c l|}
\hline
    $\directD{\rho}{\Gamma}{Y}$(!foldl (x,y.e$_1$) e$_2$ e$_3$!) &=&
            !let v,Y$_1$ = $\directD{\rho}{\Gamma}{Y}$(e$_2$) in! \\
            && !let A,Y$_2$ = $\directD{\rho}{\Gamma,v}{Y_1}$(e$_3$) in! \\
            && !let A$_0$,r$_1$ = (scanl (x,y.e$_1$) v A) in! \\
            && !let A$_1$ = map2 (a,b.(!$\grad_{\{x\}}$!e$_1$)[a/x,b/y]) A$_0$ A! \\
            && !let A$_2$ = map2 (a,b.(!$\grad_{\{y\}}$!)e$_1$[a/x,b/y]) A$_0$ A! \\
            && !let r$_2$, A$_3$ = scanr * 1 A$_1$! \\
            && !<r$_1$, fun (x$_1$,$\ldots$,x$_n$,z) ->! \\
            && !let y$_1$,B = r$_2$*z ,map2 (x,y. x*y*z) A$_2$ A$_3$ in! \\
            && !Y$_2$(x$_1$,$\ldots$,x$_n$,y$_1$,B)>! \\
    $\directD{\rho}{\Gamma}{Y}$(!map (x.e$_{1}$) e$_{2}$!) &=&  
            !let A,Y$_{1}$ = $\directD{\rho}{\Gamma}{Y}$(e$_{2}$) in! \\
            && !<map (x.e$_{1}$) A, fun (x$_{1}$,$\ldots$,x$_n$,Z) -> !\\
            && !let G = (map2 * Z!\\
            && !  (map (a.(!$\grad_\Gamma$!e$_{1}$)[a/x]) A)) in! \\
            && !Y$_{1}$( (x$_{1}$,$\ldots$,x$_n$)$\widehat{+}$(reduce $\widehat{+}$ $\widehat{0}$ G),!\\
            && \quad\quad !map2 (a,b.($\grad_{\{x\}}$e$_{1}$)[a/x]*b) A Z )>!\\\hline
\end{tabular}
\end{center}
\caption{The reverse-mode AD transformation of \codekw{foldl} and \codekw{map} operators.}
\label{fig:foldl_map_trans}
\end{figure}

\subsection{Adding non-smooth scalar operators} 
\label{sub:lift_non_smooth}

We assumed the unary and binary operators were denoted by smooth functions $\RR^n\to\RR$. 
There is no additional difficulty in considering operators which are partial functions 
like division or operators which are not smooth at a point like square root.

These functions are then given intentional derivatives which provide valid derivatives 
on the domain of definition and differentiability of the operator. 
These functions are well known to be the bete noire of AD \cite{griewank2008evaluating} 
and we do not provide novel solutions to these.  
Several recent work have shown how to give semantics to such operators in the context of AD \cite{vakar2020denotational,mazza2021automatic,sherman2021,lee2020correctness}.

\subsection{General arrays} 
\label{sub:lift_gen_arr}

We now show how to generalize our reverse-mode transformation to be defined on arrays over any ground type !G!.
That is, we need to adapt the reverse derivatives of !map2! and !reduce! when they have more general function arguments.

A ground type !G! is interpreted as an Euclidean space $A$. 
It is in particular a real vector space.
Similarly, a ground context !$\Gamma$=x$_1$:G$_1$,$\ldots$,x$_n$:G$_n$! is interpreted as $\bigoplus_{1\leq i\leq n}A_i$, where $A_i\defeq\seml$(!G$_i$!)$\semr$.
The denotation of the gradient of a term !$\Gamma \vdash$ e: G! at a point is then a matrix, more precisely an element of $(\bigoplus_{1\leq i\leq n}A_i)\otimes A$
where $\otimes$ is the tensor product of real vector spaces. This space is isomorphic to $\bigoplus_{1\leq i\leq n}A_i\otimes A$.

We can define $\otimes$ on the types of our language inductively by

\begin{tabular}{r c l}
    $\reals \otimes A$ & $\defeq$ & $A$ \\
    $(A_1 \times \ldots \times A_n)\otimes A$ & $\defeq$ & $(A_1\otimes A) \times \ldots \times (A_n \otimes A)$ \\
    $\Array{A_1}{n} \otimes A$ & $\defeq$ & $\Array{A_1\otimes A}{n}$
\end{tabular}

With this definition, we recover that the gradient of $\Gamma \vdash$!e: !$\reals$ is a tuple of type $A_1\times\ldots\times A_n$ as expected.
For !map2!, we need a generalization of $*:\reals\to\reals$. 
If !e$_1$: A!, then $\grad_{\{x\}}$!e$_1$: $A\otimes A$! and we need a new primitive $\widehat{*}:(A\otimes A) \times A \to A$.
If we represent $A\otimes A$ as a matrix, then $\widehat{*}$ is matrix-vector multiplication.
Similarly, for !A$_{3}$! in !reduce! we need a new primitive $\widetilde{*}:(A\otimes A)\times(A\otimes A) \to (A\otimes A)$.
$\widetilde{*}$ corresponds to matrix-matrix multiplication.
Then, using this notation, there are only minimal changes to the reverse derivatives of !map2! and !reduce!, as can be seen in Figure~\ref{fig:lift_gen_arr}.

\begin{figure}
\begin{center}
\begin{tabular}{|r c l|}
\hline
    $\directD{\rho}{\Gamma}{Y}$(!map2 (x,y.e$_1$: G) e$_2$ e$_3$!) &=&  
    !let A,Y$_1$ = !$\directD{\rho}{\Gamma}{Y}$!(e2) in! \\
    && !let B,Y$_2$ = !$\directD{\rho}{\Gamma,A}{Y1}$!(e$_3$) in! \\
    && !<map2 (x,y.e$_1$) A B, fun (x$_1$,$\ldots$,x$_n$,Z) -> !\\
    && !let G = (map2 $\widehat{*}$ Z!\\
    && !  (map (a.(!$\grad_\Gamma$!e$_{1}$)[a/x]) A)) in! \\
    && !Y$_2$( (x$_1$,$\ldots$,x$_n$)!$\widehat{+}$!(reduce $\widehat{+}$ $\widehat{0}$ G),!\\
    && \quad\quad! map2 (a,b.(!$\grad_{\{x\}}$!e$_1$)[a/x]!$\widehat{*}$!b) A Z,!\\
    && \quad\quad! map2 (a,b.(!$\grad_{\{y\}}$!e$_1$)[a/x]!$\widehat{*}$!b) B Z )>!\\
    $\directD{\rho}{\Gamma}{Y}$(!reduce (x,y.e$_{1}$) e$_{2}$ e$_{3}$!) &=&
    !let y$_{1}$,Y$_{1}$ = !$\directD{\rho}{\Gamma}{Y}$!(e$_{2}$) in! \\
    && !let A,Y$_{2}$ = !$\directD{\rho}{\Gamma,y_1}{Y_1}$!(e$_{3}$) in! \\
    && !let A$_{0}$ = shift1R (scanl (x,y.e$_{1}$) y$_{1}$ A) in! \\
    && !let A$_{1}$ = shift1L (map2! \\ 
    && !      (a,b.(!$\grad_{\{x\}}$!e$_{1}$)[a/x,b/y]) A$_{0}$ A) in! \\
    && !let A$_{2}$ = map2! \\
    && !      (a,b.(!$\grad_{\{y\}}$!)e$_{1}$[a/x,b/y]) A$_{0}$ A in! \\
    && !let A$_{3}$ = scanr $\widetilde{*}$ 1 A$_{1}$ in! \\
    && !<reduce (x,y.e$_{1}$) y$_{1}$ A, fun (x$_{1}$,$\ldots$,x$_n$,z) ->! \\
    && !Y$_{2}$(x$_{1}$,$\ldots$,x$_n$, map2 (x,y. x$\widehat{*}$(y$\widehat{*}$z)) A$_{2}$ A$_{3}$>! \\ \hline
\end{tabular}
\end{center}
\caption{The reverse-mode AD transformation of \codekw{map2} and \codekw{reduce} for general arrays.}
\label{fig:lift_gen_arr}
\end{figure}

Evidently, one can combine all the generalizations from the previous subsections.
Even though this transformation has the correct complexity, it is open for future research to find
even better representations to allow for more optimizations. 
In particular, representations looking like Einsum \cite{van2011numpy} could be of interest 
and has been recently studied in the context of AD \cite{laue2018computing,laue2020simple}.
More generally, there is growing interest in tensor calculus \cite{liao2019differentiable,bernstein2020differentiating}.

\subsection{Gradient from the introduction}
\label{sub:gradintro}

We show that the gradients from Section \ref{sec:intro} are obtained as instances of our general construction. 
The proofs consist in instantiating the general derivatives to these cases 
and showing that each rewrite step is a simple known optimization.

Similarly to numpy, we use the notation OnesLike(A) to mean map (x -> 1) A 
and ZerosLike(A) to mean map (x -> 0) A. 
Finding these constant arrays is key to a lot of optimizations that leverage the ring algebraic structure of the reals to arrays.

 \begin{lemma}
    !$\nabla_A$prod(A)= map2 * (scanr * 1 (shift1L A)) (shift1R (scanl * 1 A))!
 \end{lemma}

 \begin{proof}
The gradient of  derivative $\nabla_A$!prod(A)! is given by

\begin{tabular}{c l r}
    & $\nabla_A$ !prod(A)! & \\
    =&  $\nabla_A$(!reduce * 1 A!) & \\
    $\defeq$ & \Big(!fun z ->! & \\ 
    & !let A$_0$ = shift1R (scanl * 1 A) in! & \\
    & !let A$_1$ = shift1L (map2 (x,y -> y) A$_0$ A) in! & \\ 
    & !let A$_2$ = map2 (x,y -> x) A$_0$ A in! & \\
    & !let A$_3$ = scanr * 1 A$_1$ in! &\\
    & !map2 (a,b -> a*b*z) A$_3$ A$_2$! & \\
    & \Big)(1)\\
     $\stackrel{\beta-reduction}{=}$  & !let A$_0$ = shift1R (scanl * 1 A) in! & \\
    & !let A$_1$ = shift1L (map2 (x,y -> y) A$_0$ A) in! & !A$_1$=shift1L A! \\
    & !let A$_2$ = map2 (x,y -> x) A$_0$ A in! & !A$_2$= A$_0$!\\
    & !let A$_3$ = scanr * 1 A$_1$ in! & \\
    & !map2 (a,b -> a*b*1) A$_3$ A$_2$! & !map2 * A$_3$ A$_2$! \\
    = & !let A$_0$ = shift1R (scanl * 1 A) in! & \\
    & !let A$_1$=shift1L A! &\\
    & !let A$_3$ = scanr * 1 A$_1$! & \\
    & !map2 * A$_3$ A$_0$! & forward substitution !A$_2$! \\
    $\stackrel{\eta-reduction}{=}$ & !map2 * (scanr * 1 shift1L A)! &\\
    & \quad\quad\quad\quad !(shift1R (scanl * 1 A))! &
\end{tabular}

 \end{proof}

 \begin{lemma}
     $\nabla_A$!sum(A)! = !map (x -> 1) A!
 \end{lemma}

 \begin{proof}
The gradient of !sum(A)! is given by

    \begin{tabular}{c l r}
    & $\nabla_A$ !sum(A)! & \\
    = & $\nabla_A$ !reduce + 0 A! & \\
    $\defeq$ & \Big(!fun z ->! & \\ 
    & !let A$_0$ = shift1R (scanl + 0 A) in! & \\
    & !let A$_1$ = shift1L! & \\
    & \quad\quad !(map2 (x,y -> 1) A$_0$ A) in! & \\ 
    & !let A$_2$ = map2 (x,y -> 1) A$_0$ A in! & \\
    & !let A$_3$ = scanr * 1 A$_1$ in! &\\
    & !map2 (a,b -> a*b*z) A$_3$ A$_2$! & \\
    & \Big)(1)\\
    $\stackrel{\beta-reduction}{=}$  & !let A$_0$ = shift1R (scanl + 0 A) in! & \\
    & !let A$_1$ = shift1L! & !A$_1$=OnesLike(shift1L(A))! \\
    & \quad\quad !(map2 (x,y -> 1) A$_0$ A) in! & \\
    & !let A$_2$ = map2 (x,y -> 1) A$_0$ A in! & !A$_2$=OnesLike(A)!\\
    & !let A$_3$ = scanr * 1 A$_1$ in! & \\
    & !map2 (a,b -> a*b*1) A$_3$ A$_2$! & !map2 * A$_3$ A$_2$! \\
    = & !let A$_0$ = shift1R (scanl + 0 A) in! & \\
    & !let A$_1$ = OnesLike(shift1L(A))! & forward substitution \\
    & !let A$_2$ = OnesLike(A)! & forward substitution \\
    & !let A$_3$ = scanr * 1 A$_1$ in! & \\
    & !map2 * A$_3$ A$_2$! & \\
    = & !let A$_3$ = scanr * 1! & !A$_3$=OnesLike(A)!\\
    &  \quad\quad!OnesLike(shift1L(A)) in! & \\
    & !map2 * A$_3$ OnesLike(A)! & forward substitution \\
    = & !map2 * OnesLike(A) OnesLike(A)! & \\
    = & !OnesLike(A)!
    \end{tabular}

\end{proof}

 \begin{lemma}
     $\nabla_A$!dot(A,B)! = !B! 
 \end{lemma}

 \begin{proof}
The reverse derivative of !map2 * A B! is given by 

\begin{tabular}{c l}
    & !fun (Z) ->! \\
    & !let C$_1$ = map2 (a,b -> b) A B in! \\
    & !let C$_2$ = map2 (a,b -> a) A B in! \\
    & !(map2 * C$_1$ Z, map2 * C$_2$ Z)!
\end{tabular}

Let's call this term !Y!.
For convenience, let us also rewrite \\
!map2 * B (map (x -> 1) A)! 
as !let C = map (x -> 1) A in map2 * B C!.

Then the gradient of !dot(A,B)! is given by

\begin{tabular}{c l r}
    & $\nabla$ !dot(A,B)! & \\
    = & $\nabla$ !let C = map (x -> 1) A in map2 * B C! & \\
    $\defeq$ & \Big(!fun z ->! & \\ 
    & !let A$_0$ = shift1R (scanl + 0 A) in! &\\
    & !let A$_1$ = shift1L (map2 (x,y -> 1) A$_0$ A) in! & \\ 
    & !let A$_2$ = map2 (x,y -> 1) A$_0$ A in! & \\
    & !let A$_3$ = scanr * 1 A$_1$ in! &\\
    & !Y(map2 (a,b -> a*b*z) A$_3$ A$_2$)! & \\
    & \Big)!(ZerosLike(A),1)!\\
    = & !Y(OnesLike(A))! & (same reduction as previously)\\
    = & !let C$_1$ = map2 (a,b -> b) A B in! & !C$_1$ = B! \\
    & !let C$_2$ = map2 (a,b -> a) A B in! & !C$_2$ = A! \\
    & !(map2 * C$_1$ OnesLike(A),! &\\ 
    & !map2 * C$_2$ OnesLike(A))!  &\\
    = & !let C$_1$=B in! & forward substitution \\
    & !let C$_2$=A in! & forward substitution \\
    & !(map2 * C$_1$ OnesLike(A),! & !C$_1$! \\
    & !map2 * C$_2$ OnesLike(A))! & !C$_2$! \\
    = & !(map2 * B OnesLike(A),! & \\
    & !map2 * A OnesLike(A))! & \\
    = & !(B, A)!
\end{tabular}

If we are only interested in the gradient w.r.t. !A!, this indeed gives !B!.
\end{proof}

\end{document}